\documentclass{article}%
\usepackage{amsmath}
\usepackage{amsfonts}
\usepackage{mathtools}
\usepackage{enumerate}
\usepackage{amssymb}
\usepackage{graphicx}
\usepackage{bbm}%
\providecommand{\U}[1]{\protect\rule{.1in}{.1in}}
\newtheorem{theorem}{Theorem}

\newtheorem{corollary}[theorem]{Corollary}

\newtheorem{lemma}[theorem]{Lemma}

\newtheorem{proposition}[theorem]{Proposition}

\newenvironment{proof}[1][Proof]{\noindent\textbf{#1.} }{\ \rule{0.5em}{0.5em}}
\begin{document}

\title {Blow-up rate estimates for the solutions of the bosonic
Boltzmann-Nordheim equation}

\author{J. Bandyopadhyay\thanks{University of California, Davis, Department of Mathematics, One Shields Avenue,
Davis, CA 95616, USA. E-mail: jogiab@gmail.com}, J. J. L. Vel\'{a}zquez\thanks{Institute for Applied Mathematics, University of Bonn,
Endenicher Allee 60, 53115 Bonn, Germnay. E-mail: velazquez@iam.uni-bonn.de}}

\maketitle

\begin{abstract}
In this paper we study the behavior of a class of mild solutions of the homogeneous and isotropic bosonic Boltzmann-Nordheim equation near the blow-up. We obtain some estimates on the blow-up rate of the solutions and prove that, as long as a solution is bounded above by the critical singularity $\frac{1}{x}$ (the equilibrium solutions behave like this power law near the origin), it remains bounded in the uniform norm. In the last section of the paper, we also prove a local existence result for a class of measure-valued mild solutions, which allows us to solve the Boltzmann-Nordheim equation for some classes of unbounded densities.

\end{abstract}

\section{Introduction}

\bigskip

The quantum Boltzmann equation or Boltzmann-Nordheim equation describes the
dynamics of a dilute gas of quantum particles.

In the bosonic, spatially homogeneous case, the Boltzmann-Nordheim equation reads:%

\begin{align}
\partial_{t}F_{1}  &  =\int_{\mathbb{R}^{3}}\int_{\mathbb{R}^{3}}%
\int_{\mathbb{R}^{3}}q\left(  F\right)  \mathcal{M}d^{3}p_{2}d^{3}p_{3}%
d^{3}p_{4}\ \ ,\ \ p_{1}\in\mathbb{R}^{3}\ \ ,\ \ t>0\label{E2}\\
F_{1}\left(  0,p\right)   &  =F_{0}\left(  p\right)  \ \ ,\ \ p_{1}%
\in\mathbb{R}^{3} \label{E2b}%
\end{align}

\begin{align}
q\left(  F\right)   &  =q_{3}\left(  F\right)  +q_{2}\left(  F\right)
\ \ \ \ ,\ \ \ \ \ \epsilon=\frac{\left\vert p\right\vert ^{2}}{2}%
\label{T4E1a}\\
\mathcal{M}  &  =\mathcal{M}\left(  p_{1},p_{2};p_{3},p_{4}\right)
=\delta\left(  p_{1}+p_{2}-p_{3}-p_{4}\right)  \delta\left(  \epsilon
_{1}+\epsilon_{2}-\epsilon_{3}-\epsilon_{4}\right)  \label{T4E1b}%
\end{align}%
\begin{align}
q_{3}\left(  F\right)   &  =F_{3}F_{4}\left(  F_{1}+F_{2}\right)  -F_{1}%
F_{2}\left(  F_{3}+F_{4}\right) \label{Q1E1}\\
q_{2}\left(  F\right)   &  =F_{3}F_{4}-F_{1}F_{2} \label{Q1E2}%
\end{align}
where we use the notation $F_{j}=F\left(  t,p_{j}\right)  ,\ j\in
\mathbb{R}^{3}.$ The mass of the particles is normalized to one.

\bigskip

If the distributions are isotropic in addition to being spatially homogeneous, the Boltzmann-Nordheim equation can be
written, after a suitable change in the time scale, as:%
\begin{align}
&  \partial_{t}f\left(  x\right)  +f\left(  x\right)  \int_{0}^{\infty}%
\int_{0}^{\infty}dydzW\left(  x,y,z,w\right)  f\left(  w\right)  \left(
f\left(  y\right)  +f\left(  z\right)  +1\right)  \boldsymbol{1}_{\left\{
\left(  w=y+z-x\right)  \geq0\right\}  }\label{Is1E1}\\
&  =\int_{0}^{\infty}\int_{0}^{\infty}dydzW\left(  x,w,y,z\right)  f\left(
y\right)  f\left(  z\right)  \left(  f\left(  x\right)  +f\left(  w\right)
+1\right)  \boldsymbol{1}_{\left\{  \left(  w=y+z-x\right)  \geq0\right\}
},\ \nonumber
\end{align}
where
\begin{equation}
W\left(  x,w,y,z\right)  =\frac{\min\left\{  \sqrt{x},\sqrt{w},\sqrt{y}%
,\sqrt{z}\right\}  }{\sqrt{x}} \label{IsE3}%
\end{equation}
and we have dropped the dependence of $f$ and related functions in $t$ for notational
simplicity. 

The function $f$ is not the density of particles in the space of energy. Such
density is given, up to an irrelevant multiplicative constant, by:%
\[
g\left(  x\right)  =\sqrt{x}f\left(  x\right)
\]

Then (\ref{Is1E1}), (\ref{IsE3}) become:%
\begin{align}
&  \partial_{t}g\left(  x\right)  +g\left(  x\right)  \int_{0}^{\infty}%
\int_{0}^{\infty}dydz\frac{\Phi\left(  x,y,z,w\right)  }{\sqrt{xw}}g\left(
w\right)  \left(  \frac{g\left(  y\right)  }{\sqrt{y}}+\frac{g\left(
z\right)  }{\sqrt{z}}+1\right)  \boldsymbol{1}_{\left\{  \left(
w=y+z-x\right)  \geq0\right\}  }\label{S2E1}\\
&  =\int_{0}^{\infty}\int_{0}^{\infty}dydz\Phi\left(  x,w,y,z\right)
\frac{g\left(  y\right)  g\left(  z\right)  }{\sqrt{yz}}\left(  \frac{g\left(
x\right)  }{\sqrt{x}}+\frac{g\left(  w\right)  }{\sqrt{w}}+1\right)
\boldsymbol{1}_{\left\{  \left(  w=y+z-x\right)  \geq0\right\}  }\nonumber
\end{align}%
\begin{equation}
\Phi\left(  x,w,y,z\right)  =\min\left\{  \sqrt{x},\sqrt{w},\sqrt{y},\sqrt
{z}\right\}  \label{S2E2}%
\end{equation}

The stationary or equilibrium solutions of (\ref{Is1E1}), (\ref{IsE3}) with
finite mass are
the Bose-Einstein distributions:%
\begin{equation}
g_{BE}\left(  x\right)  =m_{0}\delta\left(  x\right)  +\frac{\sqrt{x}}%
{\exp\left(  \beta x+\alpha\right)  -1}\ \label{S2E3}%
\end{equation}
where $m_{0}\geq0,$ $\beta\in\left(  0,\infty\right]  ,$ $0\leq\alpha<\infty$
and $\alpha\cdot m_{0}=0.$ 
In \cite{Lu1} and \cite{Lu2} weak solutions of (\ref{Is1E1}) were defined and studied: it was shown that conservative (i.e. conserving both mass and energy) weak solutions of (\ref{Is1E1}) exist and converge to the physically expected equilibrium distributions. More recently, in \cite{EV1} and \cite{EV3} mild solutions of (\ref{Is1E1}) were considered and well-posedness results were proved for these.

\bigskip

A remarkable feature of the stationary solutions of (\ref{S2E3}) described above is the
possibility of having a positive macroscopic fraction of particles with energy
$x=0.$ The onset of such macroscopic fraction of particles is the phenomenon
of Bose-Einstein condensation. We will term as Bose-Einstein condensate any
nonzero Dirac measure appearing in $g$ at $x=0.$

It has been conjectured in \cite{LLPR} and \cite{ST2} that the solutions of
(\ref{Is1E1}), (\ref{IsE3}) might blow-up in finite time, and that such
blow-up would mark the onset of Bose-Einstein condensate, or more precisely
in this setting, a macroscopic fraction of particles with zero energy. This
conjecture was supported by numerical simulations which supported also the
picture of self-similar behaviour for the solutions at the onset of the
singularity. Rigorous results concerning singularity formation in finite time
have been obtained in \cite{EV1}, \cite{EV2}. The main goal of this paper is to derive some quantitative
information about the behaviour of the solutions of (\ref{Is1E1}),
(\ref{IsE3}) near a blow-up point. \bigskip

\bigskip

Associated with (\ref{Is1E1}), (\ref{IsE3}) are several power laws which have some physical significance and have been considered in the mathematical and physical literature.
These power laws describe several possible asymptotic behaviours of the solutions of (\ref{Is1E1}) near the origin, i.e., as $x\rightarrow0$ . We briefly review them below.

\bigskip

We first describe three power laws associated with the steady-state solutions of (\ref{Is1E1}).  While one of them gives the asymptotic behavior of
stationary solutions near the origin, the other two power laws are associated with non-zero fluxes of mass and energy to or from the origin. In other words, each of these latter two power laws is associated with one of the two conservation laws (for mass and energy) of (\ref{Is1E1}). The corresponding exponents are called the Kolmogorov-Zakharov exponents and they are related to Kolmogorov theory of wave turbulence \cite{DNPZ}.

The power law $f\left(  x\right)  \approx\frac{1}{x}$ yields the asymptotics
of the stationary solutions (\ref{S2E3}) as $x\rightarrow0$ if $\alpha=0.$ It
has been indicated also in \cite{Spohn} that this asymptotic should give the
behaviour of $f$ solution of (\ref{Is1E1}), (\ref{IsE3}) in the presence of a
condensate. This power law describes the local behaviour of the non condensed
part of a distribution $f$ which is at equilibrium with a condensate.

The power law $f\left(  x\right)  \approx\frac{1}{x^{\frac{7}{6}}}$ describes
particle distributions yielding a particle flux towards $x=0 $ and the exponent $-\frac{7}{6}$ is associated to the mass
conservation \cite{DNPZ}. Singular solutions behaving like this power law were investigated in \cite{EMV1} and  \cite{EMV2}. In \cite{EMV2}, it was proved that solutions with initial data behaving like $x^{-7/6}$ exist locally and preserve the singularity. The total mass is not conserved in this case and one could
possibly think of the mass going into the formation of the condensate, without any back reaction .


The Kolmogorov-Zakharov power law associated to the energy conservation law is
$f\left(  x\right)  \approx\frac{1}{x^\frac{3}{2}}.$ However, this exponent is also the
critical exponent for which the number of particles per unit volume becomes
divergent. A consequence of this is that several integrals in (\ref{Is1E1}),
(\ref{IsE3}) become divergent and it it is not clear in which sense solutions
of this equation with this asymptotic behaviour as $x\rightarrow0$ are meaningful
solutions to (\ref{Is1E1}), (\ref{IsE3}).  So it is not understood in which
sense solutions with this behaviour can be given the meaning of solutions
yielding a nonzero flux of energy from or towards $x=0.$

An interesting power law behaviour is $f\left(  x\right)
\approx\frac{1}{x^{\nu}},$ $\ $with $\nu=1.234.$ This exponent has been
computed numerically. It corresponds to the power law which is generated at
the blow-up time for the distribution $f,$ taking as initial data a bounded
distribution $f_{0}$ with finite mass and energy. This exponent is different
from the previous ones and it cannot be obtained by means of dimensional
considerations, as it is the case of the other exponents.  Actually this exponent appears as the solution of a nonlinear
eigenvalue problem, as it is usually the case in problems with self-similarity
of the second kind (cf. \cite{LLPR}).

It has been argued often in the physical literature that at the onset of the singularity the exponent $\nu$ should be larger than $\frac{7}{6}$, since otherwise there would not be particle-fluxes towards the origin. However mathematically this has not been proven yet and the role of the exponent $-\frac{7}{6}$ in the asymptotic
behavior of the solutions at the blow-up time near the origin is still not clear.

\bigskip
It has been proved in \cite{EV1}, \cite{EV2} that a class of solutions of (\ref{Is1E1}), (\ref{IsE3})
with bounded initial data blow up in finite time, leading to condensation for supercritical solutions. In \cite{Lu4} it has been shown that if the initial datum is singular enough at the origin, then
a condensate is present at every later time, even if initially there is no condensate. A critical role is played by the power law $f(x)=\frac{1}{x}$. The main result in \cite{Lu3} implies that if the initial datum $f_0(x)=f(x,0)$ behaves like $\frac{1}{x^{\beta}}$ with $\beta>1$, then there is condensation at all times $t>0$, even if there is none initially.  In \cite{Lu3}
more information has been obtained about the role of the power law $\frac{1}{x}$ in condensation and a regularity result has been proved for solutions bounded above by this critical power law.  In this paper, our main goal is to further clarify its role in the blow-up of mild solutions that are initially bounded. As noted in \cite{Lu3}, if the $L^1$-norm of a solution is bounded then that implies that the uniform norm of the solution is bounded. We include a different proof of this for a class of solutions with weighted supremum norms in this paper. We then obtain, as a corollary, a lower blow-up rate for the $L^1$-norm. Our next result shows that if the initial datum is bounded, then as long as the solution is bounded above by the critical singularity $\frac{1}{x}$, it remains bounded in the uniform norm. In the end we include a local existence theorem for measure-valued mild solutions that have
bounded variation in compact intervals in an exponentially-weighted space.

\bigskip
Our results are presented in the two following sections, section $2$ contains our results and estimates relating to the blow-up of solutions of (\ref{Is1E1}), while section $3$ contains a local existence result for a class of $L^1$-data for (\ref{Is1E1}), as follows:

(2) We prove that as long as the $L^{1}$ norm of $f$
remains bounded, any solution of the BN equation which is initially bounded in a suitably weighted supremum norm ,
remains bounded. This implies that near the blow-up the $L^{1}$  must become unbounded locally in some neighborhood of the origin.
This gives, as a corollary, a Gronwall-type lower estimate on how this localized $L^1$ norm blows up in time. We then prove, as an implication of a result proved in \cite{Lu4} that if initial datum is bounded, then as long as the solution is bounded above by the critical singularity $\frac{1}{x}$, it remains bounded in the uniform norm.

(3)  We study solutions of the BN equation with singular data. We prove a local existence result for solutions of 
the BN equation in a large class of measures including those that have initial density $f_0$
such that $f_{0}\left(  x\right)  \leq\frac{1}{x^{\alpha}}$ with $\alpha<1$.
More precisely, $f_{0}\in L_{loc}^{1}\left(  0,\infty\right)  $ allows us to
obtain solutions defined locally in time, which belong to a similar class. The derived solutions are
singular near the origin, but they do not have any condensate and they are conservative.

\section{Estimates near the Blow-up}

We consider the following evolution equation for a given distribution $f(x,.)$
on $\mathbb{R}_{+}$:%

\begin{align}
\label{BN} &  & \partial_{t} f(x)+ f(x)\int_{0}^{\infty}\int_{0}^{\infty} dydz
W(x,w,y,z)f(w)\left(  f(y)+f(z)+1 \right)  \mathbbm{1}(w=y+z-x\geq0)\\
&  {} & = \int_{0}^{\infty}\int_{0}^{\infty} dydz W(x,w,y,z)f(y)f(z)\left(
f(x)+f(w)+1\right)  \mathbbm{1}(w=y+z-x\geq0) ,\nonumber
\end{align}
where
\[
W(x,w,y,z)=\frac{min(\sqrt{x},\sqrt{w},\sqrt{y},\sqrt{z})}{\sqrt{x}},
\]
and we have suppressed the time-dependence in the function $f(x,t)$ for
convenience. We now define a mild solution of (\ref{BN}) as a function
$f(x,t)$ satisfying the following equation:
\begin{align}
\label{mildsoln}f(x,t)=f(x,0)\exp\left(  -\int_{0}^{t} a[f](x,s)ds\right)
+\int_{0}^{t} ds \exp\left(  -\int_{s}^{t} d\xi a[f](x,\xi)\right)  J[f](x,s),
\end{align}
where
\begin{align}
\label{adef}
a[f](x)  &  :=\int_{0}^{\infty}\int_{0}^{\infty} dydz W(x,w,y,z)f(w)\left(
f(y)+f(z)+1\right)  \mathbbm{1}(w=y+z-x\geq0)
\end{align}
\begin{align}
\label{Jdef}
J[f](x)  &  :=\int_{0}^{\infty}\int_{0}^{\infty} dydz
W(x,w,y,z)f(y)f(z)\left(  f(x)+f(w)+1\right)  \mathbbm{1}(w=y+z-x\geq0).
\end{align}

We look at mild solutions of the BN equation which are integrable and
have a suitable decay at infinity, i.e., mild solutions which are allowed to have a sub-linear singularity near the origin but sufficient decay at infinity. More precisely, we consider non-negative
functions bounded in the following weighted sup-norm:
\begin{align}
\label{wsnorm}
||f(.,t)||^{\alpha,\gamma}_{\infty}=\sup_{x\geq0} \quad x^{\alpha}(1+x)^{\gamma}f(x,t),
\end{align}
where the exponent $\gamma>0$ is chosen to be a large number (the choice is
dictated by the estimates in the theorem below) and $\alpha\in[0,1)$. We denote the mass and energy associated with such a mild solution by $m(f)$ and $e(f)$ respectively. 

Now, given any time $T$ we define the following sup-norm in time:
\begin{align}
\label{twsnorm}
\phi(T):=\sup_{0\leq t \leq T}||f(.,t)||^{\alpha,\gamma}_{\infty}%
\end{align}

The following theorem shows that solutions in these classes stay bounded in their respective weighted sup-norms mentioned above, as long as the $L^1$-norm remains bounded.

\begin{theorem}
Consider a conservative mild solution of the BN equation which is initially bounded in the weighted sup norm $||f(.,t)||^{\alpha,\gamma}_{\infty}$ defined in (\ref{wsnorm}). Then, as long as there is some time $T>0$ such that $c(T)=\sup_{0\leq t \leq T}\int_0^{\infty} f(x,t)dx$ is bounded, there exists some constant $\kappa>0$, depending only on
$T$, $c(T)$, $\alpha$, $\gamma$ and the energy $e(f)$, such that the following bound holds:
$$\phi(T)\leq\frac{1}{\kappa}\phi(0),$$
and consequently, for subsequent times one obtains, for $0\leq n<\infty$:
$$\phi(nT)\leq \frac{1}{\kappa^n}\phi(0),$$
where $\phi(T)$ is the sup norm defined in (\ref{twsnorm}).

\end{theorem}

\begin{proof}

Recalling the definition of mild solutions, we obtain the following upper
bound for $\phi(T)$:%

\begin{align*}
\phi(T) & \leq\sup_{0\leq t \leq T}\sup_{x\geq0}\left[  \exp\left(  -\int
_{0}^{t} a[f](x,s)ds\right) x^{\alpha}(1+x)^{\gamma}f(x,0)\right. \\
& +\left. x^{\alpha}(1+x)^{\gamma}\int_{0}^{t} ds \exp\left(  -\int_{s}^{t} d\xi
a[f](x,\xi)\right)  J[f](x,s) \right] \\
& = I_{1}+I_{2},
\end{align*}

where:
\begin{align*}
I_{1} & =\sup_{0\leq t \leq T}\sup_{x\geq0}\exp\left(  -\int_{0}^{t}
a[f](x,s)ds\right) x^{\alpha}(1+x)^{\gamma}f(x,0)\\
& \leq\sup_{0\leq t \leq T}\sup_{x\geq0} \quad x^{\alpha}(1+x)^{\gamma}f(x,0)\\
& =\phi(0),
\end{align*}

and $I_{2}$ is written as a sum of three terms as follows:
\begin{align*}
I_{2} & =\sup_{0\leq t \leq T}\sup_{x\geq0} \quad x^{\alpha}(1+x)^{\gamma}\int_{0}^{t} ds
\exp\left(  -\int_{s}^{t} d\xi a[f](x,\xi)\right)  J[f](x,s)\\
& \leq J_{1}+J_{2}+J_{3},
\end{align*}
where:
\begin{align*}
& J_{1}\\
& = \sup_{0\leq t \leq T}\sup_{x\geq0} \quad x^{\alpha}(1+x)^{\gamma}\int_{0}^{t} ds
\exp\left(  -\int_{s}^{t} d\xi a[f](x,\xi)\right) \int_{0}^{\infty}dy dz
W(w,x,y,z)f(x)f(y)f(z) \mathbbm{1}(y+z\geq x),
\end{align*}

\begin{align*}
& J_{2}\\
& =\sup_{0\leq t \leq T}\sup_{x\geq0} \quad x^{\alpha}(1+x)^{\gamma}\int_{0}^{t} ds
\exp\left(  -\int_{s}^{t} d\xi a[f](x,\xi)\right) \int_{0}^{\infty}dy dz
W(w,x,y,z)f(w)f(y)f(z) \mathbbm{1}(y+z\geq x),
\end{align*}

and
\begin{align*}
& J_{3}\\
& =\sup_{0\leq t \leq T}\sup_{x\geq0} \quad x^{\alpha}(1+x)^{\gamma}\int_{0}^{t} ds
\exp\left(  -\int_{s}^{t} d\xi a[f](x,\xi)\right) \int_{0}^{\infty}dy dz
W(w,x,y,z)f(y)f(z) \mathbbm{1}(y+z\geq x).
\end{align*}

In the above formulae $w=y+z-x$. We now estimate these three terms as follows:

i)%

\begin{align*}
& J_{1}\\
& = \sup_{0\leq t \leq T}\sup_{x\geq0} \quad x^{\alpha}(1+x)^{\gamma}\int_{0}^{t} ds
\exp\left(  -\int_{s}^{t} d\xi a[f](x,\xi)\right) \int_{0}^{\infty}\int
_{0}^{\infty}dy dz W(w,x,y,z)f(x)f(y)f(z) \mathbbm{1}(y+z\geq x)\\
& \leq\left( \sup_{0\leq t \leq T} ||f||_{\infty} \right)  \sup_{0\leq t \leq
T}\sup_{x\geq0} \int_{0}^{t} ds \exp\left(  -\int_{s}^{t} d\xi a[f](x,\xi
)\right) \int_{0}^{\infty}\int_{0}^{\infty}dy dz f(y) f(z)\\
& \leq\phi(T)\sup_{0\leq t \leq T}\int_{0}^{t}ds\left(  \int_{0}^{\infty}dy
f(y) \right) ^{2}\\
& = \phi(T)\sup_{0\leq t \leq T}\int_{0}^{t}ds \left( l[f](s) \right) ^{2}\\
& \leq T\phi(T)c(T)^2,
\end{align*}

where
\[
l[f](s)=\int_{0}^{\infty} dy f(y,s),
\]
and
\[
c(T)=\sup_{0\leq t \leq T}l[f](t).
\]

\bigskip
ii)
\begin{align*}
& J_{2}\\
& = \sup_{0\leq t \leq T}\sup_{x\geq0} \quad x^{\alpha}(1+x)^{\gamma}\int_{0}^{t} ds
\exp\left(  -\int_{s}^{t} d\xi a[f](x,\xi)\right) \int_{0}^{\infty}\int
_{0}^{\infty}dy dz W(w,x,y,z)f(w)f(y)f(z) \mathbbm{1}(y+z\geq x)\\
& =\sup_{0\leq t \leq T}\sup_{x\geq0} \quad2 x^{\alpha} (1+x)^{\gamma}\int_{0}^{t} ds
\exp\left(  -\int_{s}^{t} d\xi a[f](x,\xi)\right) \int_{0}^{\infty}\int
_{\frac{x}{2}}^{\infty}dy dz\\
& {}\qquad\qquad\qquad W(w,x,y,z)f(w)f(y)f(z) \mathbbm{1}(y>z)
\mathbbm{1}(y+z\geq x)\\
& \leq\phi(T)\sup_{0\leq t \leq T}\sup_{x\geq0}\quad2 x^{\alpha}(1+x)^{\gamma}\int
_{0}^{t} ds \int_{0}^{\infty}\int_{\frac{x}{2}}^{\infty} dy dz (1+y)^{-\gamma}
f(w)f(z) \mathbbm{1}(y+z\geq x)\\
& \leq\phi(T)\sup_{0\leq t \leq T}\sup_{x\geq0}\quad2^{\gamma+\alpha+1}\left(
\frac{1+x}{2+x}\right) ^{\gamma}\int_{0}^{t} ds \int_{0}^{\infty}\int
_{\frac{x}{2}}^{\infty} dy dz f(w)f(z)\\
& \leq2^{\alpha+\gamma+1}\phi(T)\sup_{0\leq t \leq T} \int_{0}^{t} ds \int
_{0}^{\infty}\int_{0}^{\infty}dz dw f(w)f(z)\\
& \leq2^{\alpha+\gamma+1}T\phi(T)(c(T))^{2}.
\end{align*}
\bigskip

iii) The term $J_{3}$ requires a bit more care in its evaluation. In fact, the
estimate in the regime of small values of $x$ will be treated differently from
the estimate in the regime of large values.

\bigskip

When $x$ is small, say $x\leq1$, we have the following estimate:
\begin{align*}
& J_{3}\\
& = \sup_{0\leq t \leq T}\sup_{0\leq x \leq1} \quad x^{\alpha}(1+x)^{\gamma}\int_{0}^{t}
ds \exp\left(  -\int_{s}^{t} d\xi a[f](x,\xi)\right) \int_{0}^{\infty}\int
_{0}^{\infty}dy dz W(w,x,y,z)f(y)f(z) \mathbbm{1}(y+z\geq x)\\
& \leq\phi(T)\sup_{0\leq t \leq T}\sup_{0\leq x \leq1}\quad x^{\alpha}(1+x)^{\gamma}
\int_{0}^{t} ds \int_{0}^{\infty} dy f(y)\left[\int_{0}^{1} dz z^{-\alpha
} + \int_1^{\infty} dz (1+z)^{-\gamma}\right]\\
&\leq 2^{\gamma}\frac{2-\alpha}{1-\alpha}Tc(T)\phi(T).
\end{align*}

\bigskip
\bigskip

On the other hand, when $x>1$ we have, for any number $\mu$ (to be chosen in
the course of the following computation),

\begin{align*}
& J_{3}\\
& \leq\sup_{0\leq t \leq T}\sup_{ x > 1} \quad x^{\alpha}(1+x)^{\gamma}\left[  2\int
_{0}^{t} ds \exp\left( -2(t-s)\sqrt{x}m(f)\right)  \frac{1}{\sqrt{x}} \left(
\int_{0}^{\mu x} dy \sqrt{y}f(y)\right) \left( \int_{\mu x}^{\infty
}(f(z)(1+z)^{\gamma})(1+z)^{-\gamma}\right) \right. \\
& +\left.  \int_{0}^{t} ds \exp\left( -2(t-s)\sqrt{x}m(f)\right)  \int
_{(1-\mu)x }^{\infty}\int_{(1-\mu)x}^{\infty}dy dz f(y)f(z)
\mathbbm{1}(y+z\geq x)\right] \\
& \leq2\phi(T)\sup_{0\leq t \leq T}\sup_{ x > 1} \quad x^{\alpha}(1+x)^{\gamma}(1+\mu
x)^{-\gamma+1}\frac{1}{\gamma- 1} \frac{m(f)}{\sqrt{x}}\int_{0}^{t}
ds\exp\left( -2(t-s)\sqrt{x}m(f)\right) \\
& + 2\sup_{0\leq t \leq T}\sup_{ x > 1} \quad x^{\alpha}(1+x)^{\gamma}\int_{0}^{t} ds
\exp\left( -2(t-s)\sqrt{x}m(f)\right)  \int_{(1-\mu)x }^{\infty}\int
_{x/2}^{\infty}dy dz f(y)f(z) \mathbbm{1}(y+z\geq x)\mathbbm{1}(y>z)\\
& \leq\phi(T)\frac{1}{\gamma- 1}\sup_{x>1}\left( \frac{1+x}{x}\right) \left(
\frac{1+x}{1+\mu x}\right) ^{\gamma-1}\\
& +2\phi(T)\sup_{0\leq t \leq T}\sup_{ x > 1} \quad(1+x)^{\gamma} \int_{0}^{t}
ds \exp\left( -2(t-s)\sqrt{x}m(f)\right)  \int_{(1-\mu)x }^{\infty}dy\left(
\frac{y}{(1-\mu)x}\right) ^{3/2}f(y)\int_{x/2}^{\infty} dz (1+z)^{-\gamma}\\
& \leq2\phi(T)\frac{1}{\gamma- 1}\left( \frac{1}{\mu}\right) ^{\gamma-
1}+2\phi(T)\sup_{0\leq t \leq T}\sup_{ x > 1}\quad te(f)\left( \frac{1}%
{(1-\mu)x}\right) ^{3/2}(1+x)^{\gamma}\frac{(1+\frac{x}{2})^{-\gamma+1}%
}{\gamma-1}\\
&\leq \frac{2.5e}{\gamma- 1}\phi(T)+2^{\alpha+\gamma+1}\frac{\gamma^{\frac{3}{2}}}{\gamma-1}\left( Te(f)\right) \phi(T),
\end{align*}

where in the above computation $m(f)$ and $e(f)$ are the mass and energy
associated with the function $f$ (and hence constant in time) and in the upper
bound in the last step we used $\mu=1 - \frac{1}{\gamma}$.

Putting all of the above estimates together we have:
\begin{align*}
\phi(T) & \leq\phi(0)+T(c(T))^{2}\phi(T)+2^{\alpha+\gamma+1}T(c(T))^{2}\phi
(T)+2^{\gamma}\frac{2-\alpha}{1-\alpha}(Tc(T))\phi(T)\\
& +2^{\alpha+\gamma+1}\frac{\gamma^{\frac{3}{2}}}{\gamma- 1}\left( Te(f)\right)
\phi(T)+ \frac{2.5e}{\gamma- 1}\phi(T),
\end{align*}
i.e.
\begin{align*}
\left[ 1-\frac{2.5e}{\gamma- 1}-\left(  2^{\alpha+\gamma+1}\frac{\gamma^{\frac{3}{2}}%
}{\gamma-1}e(f)+2^{\gamma}\frac{2-\alpha}{1-\alpha}c(T)+\left( 1+2^{\alpha+\gamma+1}\right)
(c(T))^{2}\right) T\right] \phi(T)\leq\phi(0)
\end{align*}
As long as $c(T)<\infty$, for $\gamma>2.5e+1$, we can choose the time $T$ small
enough so that:
\begin{align*}
1-\frac{2.5e}{\gamma- 1}-\left(  2^{\alpha+\gamma+1}\frac{\gamma^{\frac{3}{2}}%
}{\gamma-1}e(f)+2^{\gamma}\frac{2-\alpha}{1-\alpha}c(T)+\left( 1+2^{\alpha+\gamma+1}\right)
(c(T))^{2}\right) T=\kappa>0,
\end{align*}

which means,
\begin{align*}
\phi(T)\leq\frac{1}{\kappa}\phi(0).
\end{align*}
If we have, for any suitably small $T>0$, $c(T)<\infty$, we can iterate the
above estimate and obtain, for any $0\leq n <\infty$:
\begin{align*}
\phi(nT)\leq\frac{1}{\kappa^{n}}\phi(0).
\end{align*}
Thus solutions which are initially bounded in the weighted sup-norm defined
above remain bounded at all subsequent finite times if the $L^{1}$ norm is
bounded over some interval of time.
\end{proof}

\bigskip

We now prove an easy corollary to the above theorem, which shows that the blow-up of the $L^{1}$-norm takes place locally near the origin and gives us a lower estimate
for the time-dependence of this blow-up.

\begin{corollary}
Let $f(.,t)$ be a conservative mild solution of the BN equation with mass $m(f)$ and energy $e(f)$. Given any $\delta>0$ let $l_{\delta}(t)=\int_0^{\delta} f(x,t)dx$ be the localized $L^1$-norm of this solution. Then the $L^1$-norm of $f(.,t)$ can blow up only if this localized norm blows up.

Suppose the blow-up time for this localized norm is $T^*>0$. Then  for $t\in[0,T^*)$ the following estimate holds:
$$l^{\delta}(t)\geq \frac{1}{\sqrt{2(T^{*}-t)}} - C(m(f),\delta),$$
where $C(m(f),\delta)$ is a constant depending only on the mass $m(f)$ and $\delta$.
\end{corollary}

\begin{proof}
To see that the blow-up of the $L^{1}$
norm has to occur locally in space in some neighborhood of the origin, we just need to notice the following, :
\begin{align*}
c(T) & =\sup_{0\leq t\leq T}\int_{0}^{\infty} dy f(y,t)\\
& =\sup_{0\leq t \leq T}\left[ \int_{0}^{\delta} dy f(y,t) + \int_{\delta
}^{\infty} dy f(y,t)\right] \\
& \leq\sup_{0\leq t \leq T}\int_{0}^{\delta} dy f(y,t) + \sup_{0\leq t \leq T}
\frac{1}{\sqrt{\delta}}\int_{\delta}^{\infty} dy \sqrt{y}f(y,t)\\
& \leq c_{\delta}(T)+\frac{1}{\sqrt{\delta}}m(f),
\end{align*}

where $\delta>0$ is any fixed number and
\[
c_{\delta}(T)=\sup_{0\leq t \leq T}\int_{0}^{\delta} dy f(y,t).
\]
It is then clear that the $L^{1}$ norm can only blow up in an interval
$[0,\delta]$, for any chosen $\delta>0$, since the other part of the upper
bound is constant in time.


 Now the evolution equation for the density distribution $f(x,t)$ implies:
\begin{align*}
\frac{d}{dt}l_{\delta}(t) & \leq \int_0^{\delta}dx\int_0^{\infty}\int_0^{\infty}dydz W(w,x,y,z)f(y)f(z)\left(f(x)+f(w)+1\right)\mathbbm{1}(w=y+z-x\geq 0)\\
& \leq \int_0^{\delta}dx\int_0^{\infty}\int_0^{\infty}dydzf(y)f(z)\left(f(x)+f(w)+1\right)\\
& \leq \left[ l_{\delta}(t)+ C(m(f),\delta) \right]^3,
\end{align*}
where $$C(m(f),\delta)=\max\left(\frac{1}{\sqrt{\delta}}m(f), 1+\delta\right).$$

The blow-up time for the $L^1$ norm is $T^*$.  Then at  any time $t$ in some time-interval to the left of this critical time, say $t\in[T_0,T^*]$,
$l_{\delta}(t)$ must be a super solution of the differential equation: $$\frac{d}{dt} h(t)= \left[ h(t)+ C(m(f),\delta) \right]^3,$$ with the condition that $\lim_{t\rightarrow T^{*-}} h(t)=\infty.$
This means that the localized $L^1$ norm satisfies the following inequality:

$$l^{\delta}(t)\geq \frac{1}{\sqrt{2(T^*-t)}} - C(m(f),\delta).$$
\end{proof}

We now prove that if any solution, which is bounded initially by a constant at the origin and by the power law $1/x$ everywhere else, stays bounded for some time.  This  yields a condition about the blow-up (in unweighted sup norm) of solutions that are initially bounded in a given weighted sup-norm. This result is obtained as an implication of Theorem $5.1$ in \cite{Lu4} .

\begin{theorem}
Suppose $f(x,0)\in L^{\infty}(\mathbb{R_+}, (1+x)^{\gamma})$ has mass $m(f)<\infty$ , energy $e(f)<\infty$ and satisfies
$$f(x,0)\leq\min\left(1, \frac{1}{x}\right).$$ 
Let $T_*$ be the blow-up time for mild solutions of (\ref{BN}) in $L^{\infty}((\mathbb{R_+}, (1+x)^{\gamma}))$ having initial datum $f(x,0)$, mass $m(f)$ and energy $e(f)$ . Then:

$$\lim_{t\rightarrow T_*^{-}} \left(\sup_{x\in\mathbb{R}_+}xf(x,t)\right)=\infty.$$

\end{theorem}

\begin{proof}

Let us note that, by Thm. $5.1$ of \cite{Lu4}, under the initial conditions stated above, there exists some time $\tau>0$ such that there is a conservative mild solution $f(x,t)$, with initial datum $f(x,0)$, mass $m(f)$ and energy $e(f)$, such that  $f(x,t)\leq\frac{C(\tau)}{x}$
for all $(x,t)\in R_{+}\times[0,\tau]$, where $C(\tau)>1$ is a finite constant depending on $\tau$.


Our proof consists in the following argument:

1. We first show that, if $T>0$ is such that there exists a finite constant $C$ with $xf(x,t)<C$, for all $(x,t)\in \mathbb{R}_+\times[0,T]$, then :

\begin{eqnarray}\label{LuBound}
f(x,t)\leq C\min\left(k(T),\frac{1}{x}\right), \qquad\mbox{for all}\quad (x,t)\in\mathbb{R}_+\times[0,T],
\end{eqnarray} where $k(T)$ is a constant depending on $T$.

2. The above bound implies that $$T_* \geq \tilde{T}, \qquad \mbox{where}\qquad \tilde{T}=\sup\{\tau>0: xf(x,t)<\infty\quad  \forall  (x,t)\in \mathbb{R}_+\times[0,\tau] \}.$$

3. However, $T_*\ngtr \tilde{T},$ because, otherwise, we would be able to find a $T'$,  $\tilde{T}\leq T'<T_*$,  such that $\sup_{x\in\mathbb{R}_+} xf(x,T')=\infty$ and  $f(x,T')<\infty$.
Since $f(x,T')<\infty$, the $L^1$ norm of $f$ stays bounded over the time-interval $[0, T']$, and, by Theorem 1 of this paper, this implies the boundedness of the quantity $(1+x)^{\gamma}f(x,t) $
for all $(x,t)\in\mathbb{R}_+\times[0,T'],$ which means $\sup_{x\in\mathbb{R}_+} xf(x,T')<\infty$.

Thus $T_*=\sup\{\tau>0: xf(x,t)<\infty, (x,t)\in \mathbb{R}_+\times[0,\tau] \}$, and $\lim_{t\rightarrow T_*^{-}} \left(\sup_{x\in\mathbb{R}_+}xf(x,t)\right)=\infty.$
   
\vspace{0.2in}

Now we just need to prove (\ref{LuBound}) and the rest of the argument would follow as above.
To this end let us introduce the following comparison function: $$\Phi(x,t)= C\min\left(\lambda(t),\frac{1}{x}\right),$$ where $ \lambda(t)=\exp  \left(Ct\left(2e(f)+11C^2+2C\right)\right).$

Note that, since by our assumption, the  inequality $f(x,t)\leq\frac{C}{x}$ holds for $(x,t)\in\mathbb{R}_+\times[0,T]$ , $f(x,t)>\Phi(x,t)$ implies that $x<\lambda(t)^{-1}\leq 1$, where $t\in[0,T].$
Now, for $s\in[0,T]$ we have:

\begin{align}
\label{diffevol}
\partial_{s}\left(f(x,s)-\Phi(x,s)\right)_{+}=\left(Q(f)(x,s)-\partial_s\Phi(x,s)\right)\boldsymbol{1}(f(x,s)>\Phi(x,s)),
\end{align}
where
\begin{eqnarray*}
&&Q(f)(x,s)\\
&=&\int_{\mathbb{R}_+^2} dy dz W(w,x,y,z)\left[f(y,s)f(z,s)(f(x,s)+f(w,s)+1) - f(x,s)f(w,s)(f(y,s)+f(z,s)+1)\right]\boldsymbol{1}(w=y+z-x\geq 0)\\
&=&\int_{\mathbb{R}_+^2} dy dz W(w,x,y,z) f(y,s)f(z,s)(1+ f(w,s))\boldsymbol{1}(w=y+z-x\geq 0)-\\
&{}& - f(x,s)\int_{\mathbb{R}_+^2} dy dz W(w,x,y,z)\left(f(w,s)(f(y,s)+f(z,s)+1) - f(y,s)f(z,s)\right)\boldsymbol{1}(w=y+z-x\geq 0)\\
&=& Q_+(f)(x,s) - f(x,s)Q_-(f)(x,s).
\end{eqnarray*}
Let us note here that $Q_-(f)(x,s)\geq 0$ (cf. Lemma $2.3$,\cite{Lu4}).

Integrating (\ref{diffevol}), we have, for all $t\in[0,T]$ and a.e. in $\mathbb{R}_+$:

$$\left(f(x,t)-\Phi(x,t)\right)_{+}=\int_0^t ds\left(Q(f)(x,s)-\partial_s\Phi(x,s)\right)\boldsymbol{1}(f(x,s)>\Phi(x,s)),$$
and, consequently:
\begin{eqnarray}\label{massevol}
m\left(f(t)-\Phi(t)\right)_+ &=& \int_{\mathbb{R}_+} dx\sqrt{x}\int_0^t ds\left(Q(f)(x,s)-\partial_s\Phi(x,s)\right)\boldsymbol{1}(f(x,s)>\Phi(x,s))\\ \nonumber
&\leq&  \int_{\mathbb{R}_+} dx\sqrt{x}\int_0^t ds\left(Q_+(f)(x,s)-\partial_s\Phi(x,s)\right)\boldsymbol{1}(f(x,s)>\Phi(x,s)).
\end{eqnarray}

We now estimate first the cubic term and then the quadratic term in $Q_+(f)(x,s)$. Using Lemma $5.3$ of \cite{Lu4} we have the following upper bound:
\begin{eqnarray*}
&&\int_0^t ds\int_{\mathbb{R}_+} dx\sqrt{x}\boldsymbol{1}(f(x,s)>\Phi(x,s))\int_{\mathbb{R}_+^2} dydz W(w,x,y,z) f(w,s)f(y,s)f(z,s) \boldsymbol{1}(w=y+z-x\geq 0)\\
&\leq& \int_0^t ds\int_{\mathbb{R}_+} dx\sqrt{x}\boldsymbol{1}(f(x,s)>\Phi(x,s))\int_{\mathbb{R}_+^2} dydz W(w,x,y,z) \Phi(w,s)\Phi(y,s)\Phi(z,s) \boldsymbol{1}(w=y+z-x\geq 0)+\\
&{}&+\int_0^t ds\int_{\mathbb{R}_+} dx\sqrt{x}\boldsymbol{1}(f(x,s)>\Phi(x,s))\mathcal{J}(f,\Phi)(x,s),
\end{eqnarray*}
where
\begin{eqnarray*}
&& \int_0^t ds\int_{\mathbb{R}_+} dx\sqrt{x}\boldsymbol{1}(f(x,s)>\Phi(x,s))\mathcal{J}(f,\Phi)(x,s)\\
&=& \int_0^t ds\int_{\mathbb{R}_+} dx \boldsymbol{1}(f(x,s)>\Phi(x,s))\int_{\mathbb{R}_+^2} dydz \boldsymbol{1}(w=y+z-x\geq 0)\left(\min(w,x,y,z)\right)^{\frac{1}{2}}\left[f(w,s)f(y,s)\left(f(z,s)-\Phi(z,s)\right)_{+}+\right.\\
&{}&\left.+f(z,s)f(w,s)\left(f(y,s)-\Phi(y,s)\right)_{+} + f(z,s)f(y,s)\left(f(w,s)-\Phi(w,s)\right)_{+}\right]
\end{eqnarray*}

We first estimate the integral involving $\mathcal{J}(f,\Phi)(x,s)$, which consists of three terms as shown above. Let us look at the following integral:

$$J_1=\int_{\mathbb{R}_+} dx \boldsymbol{1}(f(x,s)>\Phi(x,s))\int_{\mathbb{R}_+^2} dydz \boldsymbol{1}(w=y+z-x\geq 0)\left(\min(w,x,y,z)\right)^{\frac{1}{2}}f(w,s)f(y,s)\left(f(z,s)-\Phi(z,s)\right)_{+}.$$
We estimate this term in the different regions of the $y-z$-plane as follows:

Suppose $\min(w, x, y, z)=z$. Let us call this region $\Delta_4.$ In this region $z\leq w=y+z-x\leq y.$
Thus:

\begin{eqnarray*}
&&J_1(\Delta_4)\\
&=&\int_{\mathbb{R}_+} dx \boldsymbol{1}(f(x,s)>\Phi(x,s))\iint\limits_{\Delta_4} dydz\sqrt{z}f(w,s)f(y,s)\left(f(z,s)-\Phi(z,s)\right)_{+}\\
&\leq&\int_{\mathbb{R}_+} dx \boldsymbol{1}(f(x,s)>\Phi(x,s))\iint\limits_{\Delta_4} dydz\sqrt{z}f(w,s)\left((f(y,s)-\Phi(y,s))_+ + \Phi(y,s)\right)\left(f(z,s)-\Phi(z,s)\right)\\
&\leq& \int_{\mathbb{R}_+} dx \boldsymbol{1}(f(x,s)>\Phi(x,s))\iint\limits_{\Delta_4} dydz \sqrt{z}f(w,s)\left((f(y,s)-\Phi(y,s)\right)_+ f(z,s)+\\
&{}& +  \int_{\mathbb{R}_+} dx \boldsymbol{1}(f(x,s)>\Phi(x,s))\iint\limits_{\Delta_4} dydz \sqrt{z}f(w,s)\Phi(y,s) \left(f(z,s)-\Phi(z,s)\right)_+\\
&\leq& \int_0^{\infty} dy \left(f(y,s) - \Phi(y,s)\right)_+\int_0^y dw f(w,s)\int_0^w dz f(z,s)\sqrt{z} + \\
&{}& + \int_0^{\lambda(s)^{-1}} dx\int\int_{\Delta_4} dydz \sqrt{z}f(w,s)\Phi(y,s) \left(f(z,s)-\Phi(z,s)\right)_+\left(\boldsymbol{1}(y<2\lambda(0)^{-1})+ \boldsymbol{1}(y\geq2\lambda(0)^{-1})\right)\\
&\leq& 4C^2\int_0^{\infty} dy \sqrt{y}\left(f(y,s) - \Phi(y,s)\right)_+ + C\lambda(s)\int_0^{\lambda(s)^{-1}} dx\int_0^{2\lambda(0)^{-1}} dw \left(f(w,s)-\Phi(w,s)\right)_+\int_0^w dz \sqrt{z}f(z,s)+\\
&{}& + C\lambda(s)\int_0^{\lambda(s)^{-1}} dx\int_0^{2\lambda(0)^{-1}} dw \Phi(w,s)\int_0^w dz \sqrt{z}\left(f(z,s)-\Phi(z,s)\right)_+ +\\
&{}& + \int_0^{\lambda(s)^{-1}} dx C\lambda(s)\int_{\lambda(0)^{-1}}^{\infty} dw\frac{\sqrt{w}}{\lambda(0)^{-\frac{1}{2}}}f(w,s)\int_0^w dz\sqrt{z}\left(f(z,s)-\Phi(z,s)\right)_+\\
&\leq& \left(4C^2+C(\lambda(0))^{\frac{1}{2}}m(f)\right) m\left(\left(f(s)-\phi(s)\right)_+\right),
\end{eqnarray*}
where  we have used the fact that $f(x,s)\leq\frac{C}{x}$. $J_1$ is estimated in the other regions of the $y-z$ plane in a similar manner, yielding the bound:
\begin{eqnarray*}
J_1&=&\int_{\mathbb{R}_+} dx \boldsymbol{1}(f(x,s)>\Phi(x,s))\int_{\mathbb{R}_+^2} dydz \boldsymbol{1}(w=y+z-x\geq 0)\left(\min(w,x,y,z)\right)^{\frac{1}{2}}f(w,s)f(y,s)\left(f(z,s)-\Phi(z,s)\right)_{+}\\
&\leq& A_1\left(C,\lambda(s), m(f)\right) m\left(\left(f(s)-\phi(s)\right)_+\right),
\end{eqnarray*}

where $A_1$ is a constant depending on $m(f)$, $\lambda(s)$ and $C$.

In a similar manner we obtain the estimate :
\begin{eqnarray*}
J_2 &=&\int_{\mathbb{R}_+} dx \boldsymbol{1}(f(x,s)>\Phi(x,s))\int_{\mathbb{R}_+^2} dydz \boldsymbol{1}(w=y+z-x\geq 0)\left(\min(w,x,y,z)\right)^{\frac{1}{2}}f(w,s)f(z,s)\left(f(y,s)-\Phi(y,s)\right)_{+}\\
&\leq& A_2\left(C,\lambda(s), m(f)\right) m\left(\left(f(s)-\phi(s)\right)_+\right).
\end{eqnarray*}

The remaining integral

$$J_3 =\int_{\mathbb{R}_+} dx \boldsymbol{1}(f(x,s)>\Phi(x,s))\int_{\mathbb{R}_+^2} dydz \boldsymbol{1}(w=y+z-x\geq 0)\left(\min(w,x,y,z)\right)^{\frac{1}{2}}f(y,s)f(z,s)\left(f(w,s)-\Phi(w,s)\right)_{+},$$
can also be estimated in a similar way, except for the region where $\min(w, x, y, z)=y$, which we call $\Delta_2.$ Thus we write the explicit estimate for this region below. In $\Delta_2$,  $y\leq w\leq z.$
Then we have:
\begin{eqnarray*}
&&J_3(\Delta_2)\\
&\leq& \int_0^{\lambda(s)^{-1}} dx\iint\limits_{\Delta_2} dydz \sqrt{y}f(y,s)\left(f(z,s)-\Phi(z,s)\right)\left(f(w,s)-\Phi(w,s)\right)_{+}\boldsymbol{1}(z\leq 1)+\\
&{}&+\int_0^{\lambda(s)^{-1}} dx\iint\limits_{\Delta_2} dydz \sqrt{y}f(y,s)\Phi(z,s)\left(f(w,s)-\Phi(w,s)\right)_{+}\boldsymbol{1}(z\leq 1)+\\
&{}& + \int_0^{\lambda(s)^{-1}} dx\iint\limits_{\Delta_2} dydz \sqrt{y}f(y,s)f(z,s)\left(f(w,s)-\Phi(w,s)\right)_{+}\boldsymbol{1}(z >1)\\
&\leq& \int_0^1 dz \left(f(z,s)-\Phi(z,s)\right)_+ \int_0^z dw f(w,s)\int_0^w dy \sqrt{y}f(y,s) + \\
&{}&+ \int_0^1 dz \Phi(z,s)\int_0^{\infty}dw\left(f(w,s)-\Phi(w,s)\right)_+\int_0^w dy \sqrt{y}f(y,s) +\\
&{}&+ \int_1^{\infty} f(z,s)\int_0^z dw \left(f(w,s)-\Phi(w,s)\right)_+\int_0^w dy\sqrt{y}f(y,s)\\
&\leq& A_3\left(C,\lambda(s), m(f)\right) m\left(\left(f(s)-\phi(s)\right)_+\right).
\end{eqnarray*}

Putting all of the above estimates together we arrive at the following bound:

\begin{eqnarray}\label{massterm}
\int_0^t ds\int_{\mathbb{R}_+} dx\sqrt{x}\boldsymbol{1}(f(x,s)>\Phi(x,s))\mathcal{J}(f,\Phi)(x,s)\leq \int_0^t ds A\left(C,\lambda(s), m(f)\right)m\left(\left(f(s)-\phi(s)\right)_+\right),
\end{eqnarray}

where $A\left(C,\lambda(s), m(f)\right)$ is a positive function of $s$.

Let us now look at the other part of the cubic term.
\begin{eqnarray*}
&&\iint\limits_{\mathbb{R}_+^2} dydz W(w,x,y,z)\Phi(w,s)\Phi(y,s)\Phi(z,s)\boldsymbol{1}(f(x,s)>\Phi(x,s))\boldsymbol{1}(w=y+z-x\geq0)\\
&\leq& \lambda(s)\iint\limits_{\mathbb{R}_+^2} dy'dz'\boldsymbol{1}(x'<1)W(w',x',y',z')\Phi'(w')\Phi'(y')\Phi'(z')\boldsymbol{1}(w'=y'+z'-x'\geq0),
\end{eqnarray*}
where we have rescaled the energy variables and the function $\Phi$ as $w'=w\lambda(s)$,  $x'=x\lambda(s)$, $y'=y\lambda(s)$, $z'=z\lambda(s)$ and $\Phi(y,s)=\lambda(s)\Phi'(y')$, so that,
$$\Phi'(y')=C\min\left(1,\frac{1}{y'}\right),$$ and $$W(x',y',y',z')=\frac{1}{\sqrt{x'}}\left(\min(w',x',y',z')\right)^{\frac{1}{2}}.$$
We now estimate this  as follows:

\begin{eqnarray*}
&&\lambda(s)\iint\limits_{\mathbb{R}_+^2} dy'dz'\boldsymbol{1}(x'<1)W(w',x',y',z')\Phi'(w')\Phi'(y')\Phi'(z')\boldsymbol{1}(w'=y'+z'-x'\geq0)\\
&\leq& C^3\lambda(s)\iint \boldsymbol{1}(\max(y',z')<1) dy'dz' \Phi'(w')\Phi'(y')\Phi'(z')\boldsymbol{1}(w'=y'+z'-x'\geq0) +\\
&{}& + C^3 \lambda(s)\iint dy'dz'\boldsymbol{1}(\min(y',z')\leq1) \boldsymbol{1}(\max(y',z')\geq1) \Phi'(w')\Phi'(y')\Phi'(z')\boldsymbol{1}(w'=y'+z'-x'\geq0)+\\
&{}& + C^3\lambda(s)\iint dy'dz'\boldsymbol{1}(\min(y',z')\leq1)\Phi'(w')\Phi'(y')\Phi'(z') \boldsymbol{1}(w'=y'+z'-x'\geq0)\\
&\leq& C^3\lambda(s)\left(\int_0^1 dy'\right)^2 + 2\lambda(s)C^3\left(\int_1^2 dz' \int_0^1 dy' + \int_2^{\infty}dz' \int_0^1dy' \frac{1}{z'(z'-1)}\right)+\\
&{}&+C^3\lambda(s)\int_1^{\infty}\int_1^{\infty}dy'dz'\frac{1}{y'z'}\frac{1}{\sqrt{y'z'}}\\
&\leq& 11C^3\lambda(s)
\end{eqnarray*}

Let us now proceed to estimate the quadratic term.

\begin{eqnarray*}
&&\int_0^t ds \int_{\mathbb{R}_+}dx \sqrt{x} \boldsymbol{1}(f(x,s)>\Phi(x,s))\iint_{\mathbb{R}_+^2} W(x,y,z) f(y,s) f(z,s)\boldsymbol{1}(w=y+z-x\geq 0) dy dz\\
&\leq& 2 \int_0^t ds \int_{\mathbb{R}_+}dx \sqrt{x} \boldsymbol{1}(f(x,s)>\Phi(x,s))\left[\int_1^{\infty}dz\int_0^z dy W(x,y,z)\left( (f(y,s)-\Phi(y,s))_++\Phi(y,s)\right)f(z,s)\boldsymbol{1}(w=y+z-x\geq 0)\right.\\
&{}&\left.+\int_0^1dz\int_0^z dy W(x,y,z)\left( (f(y,s)-\Phi(y,s))_++\Phi(y,s)\right)\left((f(z,s)-\Phi(z,s))_++\Phi(z,s)\right)\boldsymbol{1}(w=y+z-x\geq 0) \right]\\
&\leq& 2 \int_0^t ds \int_{\mathbb{R}_+}dx \sqrt{x} \boldsymbol{1}(f(x,s)>\Phi(x,s))\left[\int_1^{\infty}dz\int_0^z dy \sqrt{\frac{y}{x}}\left(f(y,s)-\Phi(y,s)\right)_+f(z,s) + \int_1^{\infty}dz\int_0^z dy\Phi(y,s)f(z,s) +\right.\\
&{}&\left.+ \int_0^1dz\left(f(z,s)-\Phi(z,s)\right)_+\int_0^z dy\sqrt{\frac{y}{x}}f(y,s)+ \int_0^1 dz \Phi(z,s)\int_0^z dy \sqrt{\frac{y}{x}}\left(f(y,s)-\Phi(y,s)\right)_+ +\right.\\
&{}&\left.+\int_0^1 dz \left(f(z,s)-\Phi(z,s)\right)_+\int_0^z dy\sqrt{\frac{y}{x}}\Phi(y,s) + \int_0^1 dz\Phi(z,s)\int_0^z dy\Phi(y,s)\right]\\
&\leq& 2 \int_0^t ds \left[\int_0^1 dx \quad m(f) m\left((f(s)-\Phi(s))_+\right)+\int_{\mathbb{R}_+}dx\sqrt{x}\boldsymbol{1}(f(x,s)>\Phi(x,s))C\lambda(s)\int_1^{\infty}dz \quad zf(z,s) +\right.\\
&{}&\left.+ C\int_0^1 dx\int_0^1 \sqrt{z}\left(f(z,s)-\Phi(z,s)\right)_+ +\int_0^{\lambda(s)^{-1}} dx \int_0^1 dz \Phi(z,s)\int_0^z dy\sqrt{y}\left(f(y,s)-\Phi(y,s)\right)_+ \right.\\
&{}&\left.+\int_0^1 dx \int_0^1 dz\sqrt{z}\left(f(z,s)-\Phi(z,s)\right)_+ + \int_{\mathbb{R}_+} dx \sqrt{x}\boldsymbol{1}(f(x,s)>\Phi(x,s))\int_0^1 dz \quad z C\lambda(s)\Phi(z,s) \right]\\
&\leq& \int_0^t ds \left(m(f)+2C+1\right)m\left((f(s)-\Phi(s))_+\right)+2\int_0^t ds \int_{\mathbb{R}_+}\sqrt{x} dx\boldsymbol{1}(f(x,s)>\Phi(x,s))(C+e(f))C\lambda(s)
\end{eqnarray*}

Putting the above estimates for the cubic and quadratic terms together we get the following inequality:
\begin{eqnarray}\label{Gronwall}
m\left((f(t)-\Phi(t))_+\right)\leq \int_0^t ds B\left(C,\lambda(s), m(f)\right)m\left((f(s)-\Phi(s))_+\right) + \int_0^t ds\int_{\mathbb{R}_+} dx \sqrt{x}\left( (2e(f)C+ 2C^2+11C^3)\lambda(s)-\dot{\lambda}(s)\right),
\end{eqnarray}
where $B\left(C,\lambda(s), m(f)\right)$ is a positive function of $s$.

From our choice of $\lambda(t)$ it is clear that $(2e(f)C+ 2C^2+11C^3)\lambda(s)-\dot{\lambda}(s)=0$. Thus as application of Gronwall's lemma to (\ref{Gronwall}) yields that:
$$m\left((f(t)-\Phi(t))_+\right)=0.$$
This means that, for all $t\in[0,T]$, a.e. in $\mathbb{R}_+$ the following inequality holds:

$$f(x,t)\leq \Phi(x,t).$$  This means that $g(x,t)=\min(f(x,t),\Phi(x,t))$ is a mild solution of (\ref{BN}), which has initial datum $f(x,0)$, mass $m(f)$ and energy $e(f)$ and which remains bounded for all
$(x,t)\in\mathbb{R}_+\times[0,T]$.  It is then obvious that $T_*>T$, where $T_*$ is the blow-up time for a mild solution of (\ref{BN}) with initial datum $f(x,0)$, mass $m(f)$ and energy $e(f)$.
This implies, by the arguments stated in the beginning of the proof, that :

$$\lim_{t\rightarrow T_*^{-}} \left(\sup_{x\in\mathbb{R}_+}xf(x,t)\right)=\infty.$$

\end{proof}
\vspace{0.15in}

\section{Local Existence result for a class of $L^1$-data}



In this section we present a local existence theorem for a class of measure-valued $L^1$-data, which is not related to the blow-up results stated above but of independent interest since this class of solutions of the BN equation has not been considered in previous papers.

We consider mild solutions of the BN equation, as defined by (\ref{mildsoln}).

Let us also define the following quantity:%

\begin{align*}
I[f](x,t):= \int_{0}^{t} ds \exp\left(  -\int_{s}^{t} d\xi a[f](x,\xi)\right)
J[f](x,s).
\end{align*}
We now define the following set of measures on $\mathbb{R}_+$:
$$G^{\beta}_{\kappa}=\{\mu|\mu\geq 0, \sup_{R>0}\int_{(R, R+1)}dxe^{\beta x}\mu(dx)\leq \kappa \}.$$

We also use the following notation for the norm in the above definition: $$||\mu||_{G^{\beta}}=\sup_{R>0}\int_{(R, R+1)}dxe^{\beta x}\mu(dx).$$
Then $G^{\beta}_K\subset \mathcal{M}_{+}(\mathbb{R}_{+})$, where $\mathcal{M}_{+}(\mathbb{R}_{+})$ is the cone of positive Radon measures on $\mathbb{R}_{+}$.
We will now first prove a few properties of this set of measures and then define measure-valued mild solutions for the Boltzmann-Nordheim equation.
\begin{proposition}
The set $G^{\beta}_{\kappa}$ is weakly compact.
\end{proposition}
\begin{proof}
The measures in $G^{\beta}_{\kappa}$  are uniformly bounded. To notice this, let us define a sequence of sets $$E_K=(K,K+1)\cup(K+\frac{1}{2}, K+\frac{3}{2}), \qquad K\in\mathbb{N}.$$
Then for $\mu_j\in G^{\beta}_{\kappa}$ we have:
\begin{eqnarray*}
\mu_j(\mathbb{R}_{+})&<&\mu_j(\cup_{l=0}^{\infty}E_l)\\
&\leq&\sum_{l=0}^{\infty}\mu_j(E_l)\\
&\leq&\sum_{l=0}^{\infty}\int_{(l,l+1)}e^{-\beta x}e^{\beta x}\mu_j(dx)+\sum_{l=0}^{\infty}\int_{(l+\frac{1}{2},l+\frac{3}{2})}e^{-\beta x}e^{\beta x}\mu_j(dx)\\
&\leq& \int_{(0,1)}e^{\beta x}\mu_j(dx)+ \sum_{l=1}^{\infty}e^{-\beta l}\int_{(l,l+1)}  e^{\beta x}\mu_j(dx)  +  \sum_{l=0}^{\infty} e^{-\beta (l+\frac{1}{2})}\int_{(l+\frac{1}{2}, l+\frac{3}{2})}e^{\beta x}\mu_j(dx)\\
&\leq& C(\beta)\kappa,
\end{eqnarray*}
where $C(\beta)\leq 3$ is a constant that depends only on $\beta$.

Thus any sequence of measures in $G^{\beta}_{\kappa}$ is a uniformly bounded sequence of Radon measures and hence has a subsequence that weakly converges to a Radon measure $\mu$ in $\mathcal{M}_{+}(\mathbb{R}_+)$. Let us denote the subsequence by $\{\mu_i\}$. Then we have:
$$\lim_{i\to\infty}\int_{\mathbb{R}_+}\phi(x)\mu_i(dx)=\int_{\mathbb{R}_{+}}\phi(x)\mu(dx),$$ for all continuous, bounded functions $\phi.$
We now have to show that this limiting measure $\mu\in G^{\beta}_{\kappa}.$

Given any $R> 0$ the measure $\nu(dx)=e^{\beta x}\mathbbm{1}(R\leq x\leq R+1)\mu(dx)$ is inner regular. Now let us construct the following function on $\mathbb{R}_{+}$, for $m\geq 3$:
\begin{eqnarray*}
\psi_m(x)&=&\sin\left(\frac{m\pi}{2}(x-R)\right)e^{\beta x},\qquad R< x< R+\frac{1}{m}\\
&=& e^{\beta x},  \quad\qquad R+\frac{1}{m}\leq x\leq R+1-\frac{1}{m}\\
&=&\sin\left(\frac{m\pi}{2}(R+1-x)\right)e^{\beta x},\qquad R+1-\frac{1}{m}< x< R+1\\
&=& 0\qquad \text{otherwise.}
\end{eqnarray*}
Also, let $A_m=[R+\frac{1}{m}, R+1-\frac{1}{m}].$
Then we have:

\begin{eqnarray*}
\int_{(R,R+1)}e^{\beta x}\mu(dx)&=&\int_{(R,R+1)}\nu(dx)\\
&=& \sup_{A_m}\int_{A_m}\nu(dx)\\
&=& \limsup_{m\rightarrow\infty}\int \mathbbm{1}\left(R+\frac{1}{m}\leq x\leq R+1-\frac{1}{m}\right) e^{\beta x}\mu(dx)\\
&\leq&\limsup_{m\rightarrow\infty}\int \psi_{m}(x)\mu(dx).
\end{eqnarray*}

The compactly supported function $\psi_m(x)$ is continuous and bounded. Thus :$$\lim_{i\rightarrow\infty}\int\psi_m(x)\mu_i(dx)=\int\psi_m(x)\mu(dx).$$

For every $i$ we have:

\begin{eqnarray*}
\int\psi_m(x)\mu_i(dx)\leq\int_{(R,R+1)}e^{\beta x}\mu_i(dx)\leq \kappa\qquad \text{for all $m$.}
\end{eqnarray*}

This means, $$\int\psi_m(x)\mu(dx)\leq\kappa\qquad\text{for all $m$},$$
which in turn means:
$$\int_{(R,R+1)}e^{\beta x}\mu(dx)\leq\limsup_{m\rightarrow\infty}\int\psi_m(x)\mu(dx)\leq\kappa.$$
Thus, $\mu\in G^{\beta}_{\kappa}.$

\end{proof}

We now define the following uniform norm $$||\mu||_{G^{\beta},T}= \sup_{0\leq t\leq T}\sup_{R>0}\int_{(R, R+1)}dxe^{\beta x}\mu(t,dx).$$
Let us define the space:
$$G^{\beta}_{\kappa, T} = \{ f\in C \left([0, T]; G^{\beta}\right); ||f||_{G^{\beta},T} \leq \kappa \}.$$

   We can metrize $G^{\beta}_{\kappa}$ by introducing the following distance function:
$$d(\mu,\nu)=\sup_{f}\left\{|\int f d\mu - \int f d\nu|: f\in BL(\mathbb{R_{+}}), ||f||_{BL}\leq 1\right\},$$
where $$||f||_{BL}=||f||_{\infty}+\sup_{x, y, x\neq y}\frac{|f(x)-f(y)|}{|x-y|}.$$
The metric for $G^{\beta}_{\kappa, T}$ will then be:
$$\rho_{G^{\beta}_{\kappa, T}}(\mu,\nu)=\sup_{0\leq t\leq T} d(\mu,\nu).$$

In the subsequent arguments we would only need to use the fact that the space $G^{\beta}_{\kappa,T}$ is metrizable and not the metric described above.  The metrizability of $G^{\beta}_{\kappa, T}$ means that any weakly sequentially continuous operator $\mathcal{T}$ on this space mapping weakly compact subsets to weakly compact subsets has a fixed point in $G^{\beta}_{\kappa, T}$.
We also need to keep in mind that any equicontinuous family of functions $\mathcal{F}\in G^{\beta}_{\kappa, T}$ is a compact set in $G^{\beta}_{\kappa, T}$, by the Arzela-Ascoli theorem.

Let us now define measure-valued mild solutions of \eqref{BN}. To do this we first define the function $a(x,s)$ and the measure $J[f](s,.)$ that occur in the expression \eqref{mildsoln}:

For $f\in C\left([0, T], M_{+}(\mathbb{R}_{+})\right)$ we define:
\begin{eqnarray*}
a(x,s) &=& \int_{0}^{\infty}\int_{0}^{\infty}W(x,w,y,z)f(s,dw)f(s,dy)\mathbbm{1}(z=x+w-y\geq 0) \\
&+& \int_{0}^{\infty}\int_{0}^{\infty}W(x,w,y,z)f(s,dw)f(s,dz)\mathbbm{1}(y=x+w-z\geq 0) +  \int_{0}^{\infty}\int_{0}^{\infty} W(w,x,y,z)f(s,dw)dy\mathbbm{1}(z=x+w-y\geq 0).
\end{eqnarray*}
On the other hand, given $f\in C\left([0, T], M_{+}(\mathbb{R}_{+})\right)$, $J[f](s,.)$ is a measure  which is defined by its action on any test function $\phi\in C_{b}(\mathbb{R}_{+})$ as follows:
\begin{eqnarray*}
\int_{0}^{\infty}\phi(x)J[f](s,dx) &=& \int_{0}^{\infty}\int_{0}^{\infty}\int_{0}^{\infty}\phi(x)W(w,x,y,z)f(s,dx)f(s,dy)f(s,dz)\mathbbm{1}(w=y+z-x\geq0)\\
&+&\int_{0}^{\infty}\int_{0}^{\infty}\int_{0}^{\infty}\phi(x)W(w,x,y,z)f(s,dw)f(s,dy)f(s,dz)\mathbbm{1}(x=y+z-w\geq 0)\\
&+& \int_{0}^{\infty}\int_{0}^{\infty}\int_{0}^{\infty}\phi(x)W(w,x,y,z)f(s,dy)f(s,dz)dx\mathbbm{1}(w=y+z-x\geq 0).
\end{eqnarray*}

Then a measure-valued mild solution at time $t$ is:
\begin{eqnarray}\label{MMsoln}
f(t,dx)= \exp\left(-\int_{0}^t a(x,s)ds\right)f(0,dx) + \int_{0}^t ds \exp\left(-\int_{s}^t a(x,\xi)d\xi \right)J[f](s,dx),
\end{eqnarray}
and its action on a test function $\phi\in C_{b}(\mathbb{R}_{+})$ is given by:
\begin{eqnarray}\label{mesmildsoln}
\int_0^{\infty}\phi(x)f(t,dx)= \int_{0}^{\infty}\phi(x)\exp\left(-\int_{0}^t a(x,s)ds\right)f(0,dx) + \int_{0}^{\infty}\phi(x)\int_{0}^t ds \exp\left(-\int_{s}^t a(x,\xi)d\xi \right)J[f](s,dx).
\end{eqnarray}

We now define the time evolution operator $\mathcal{T}$ on $C\left([0,T], M_{+}(\mathbb{R}_{+})\right)$ as follows:
\begin{eqnarray}\label{evolnop}
\mathcal{T}[f](t,dx)= \exp\left(-\int_{0}^t a(x,s)ds\right)f(0,dx) + \int_{0}^t ds \exp\left(-\int_{s}^t a(x,\xi)d\xi \right)J[f](s,dx)
\end{eqnarray}

In order to prove the well-posedness of the problem we have to show that the operator $\mathcal{T}$ has a fixed point in $G^{\beta}_{\kappa,T})$. We will show this by proving that $\mathcal{T}$ is a weakly sequentially continuous operator mapping any compact subset of $G^{\beta}_{\kappa,T})$ to itself. This will be done by proving a few lemmas. Observe that  the mass and energy functionals associated with any measure $\mu\in G^{\beta}_{\kappa}$ are finite.  The next lemma shows that, given a sequence $\mu_j$ of measures weakly converging to $\mu\in G^{\beta}_{\kappa}$, the associated mass and energy functionals  converge to their expected limits.
\begin{lemma}
For any sequence of measures $\mu_j\in G^{\beta}_{\kappa}$ converging weakly to $\mu\in G^{\beta}_{\kappa}$ , we have:
$$\lim_{j\rightarrow\infty}\int x^{\alpha}d\mu_j=\int x^{\alpha} d\mu,$$ where $\alpha=\frac{1}{2},\frac{3}{2}.$
\end{lemma}

\begin{proof}
The limit stated in the theorem is obviously true on compact sets. Thus let us choose some large $R$ and notice the following, for any measure $\mu_i\in G^{\beta}_{\kappa}$:

\begin{eqnarray*}
\int_{(R^{\frac{1}{\alpha}},\infty)}x^{\alpha}\mu_i(dx)&\leq&\sum_{l=R^{\frac{1}{\alpha}}}^{\infty}\left[\int_{(l,l+1)}e^{\beta x}x^{\alpha}e^{-\beta x}\mu_i(dx)\right.\\
&{}&\left.+\int_{(l+\frac{1}{2}),l+\frac{3}{2})}e^{\beta x}x^{\alpha}e^{-\beta x}\mu_i(dx)\right]\\
&\leq& C(\alpha,\beta)\sum_{l=R^{\frac{1}{\alpha}}}^{\infty}\left[\int_{(l,l+1)}e^{\beta x}e^{-\beta x/2}\mu_i(dx)+\int_{(l+\frac{1}{2}, l+\frac{3}{2})}e^{\beta x}e^{-\beta x/2}\mu_i(dx)\right]\\
&\leq& C(\alpha,\beta) e^{-\beta\frac{R^{\frac{1}{\alpha}}}{2}}\kappa,
\end{eqnarray*}
where $C(\alpha,\beta)\leq 3$ is a positive constant depending only on $\alpha$ and $\beta$.
This means, $$\lim_{R\rightarrow\infty}\int_{(R^{\frac{1}{\alpha}},\infty)}x^{\alpha}\mu_i(dx)=0,\qquad\text{for all}\quad\mu_i\in G^{\beta}_{\kappa}.$$
Then convergence on compact sets means that $$\lim_{j\rightarrow\infty}\int x^{\alpha}d\mu_j=\int x^{\alpha} d\mu.$$

\end{proof}

We would now like to prove the following lemma:

\begin{lemma}
Let $G^{\beta}\subset\mathcal{M}_{+}(\mathbb{R}_{+})$ denote the space of positive Radon measures with the following norm:
\[
||f||_{G^{\beta}}=\sup_{R>0}\int_{R}^{R+1}e^{\beta x}f(dx),
\]
where $\beta>1$ depends only on the mass $m(f)=\int_{0}^{\infty}\sqrt
{x}f(dx)$ associated with the measure $f$. Then for $f\in G^{\beta}$ we have
the following estimate:
\begin{eqnarray}\label{bound}
||\mathcal{T}[f]||_{G^{\beta}}(t)\leq  ||f_0||_{G^{\beta}} + C\left(  t+\frac{1}{m(f)\sqrt{\beta}}\right)
||f||_{G^{\beta}}^{2}(1+||f||_{G^{\beta}}),
\end{eqnarray}
where $C$ is a numerical constant and $f_0=f(0,.)$.
This estimate implies in particular that there exists some time $T_{*}$ such that the operator $\mathcal{T}$ maps $G^{\beta}_{\kappa, T_{*}}$ to itself given $f(0,.)$, such that $||f_0||_{G^{\beta}}<\kappa$.

\end{lemma}

\begin{proof}
To prove the theorem let us define:
\[
I^{R}[f](t):=\int_{R}^{R+1} e^{\beta x}\int_{0}^{t} ds \exp\left(  -\int
_{s}^{t} d\xi a[f](x,\xi)\right)  J[f](dx,s).
\]
We now write the right hand side in a form that is more amenable to computations. This involves a change of variable in one of the terms in the measure $J[f](dx,s)$, so that $I^{R}[f](t)$
can be expressed as:
\begin{eqnarray*}
&&I^{R}[f](t)=\int_{R}^{R+1}e^{\beta x}\int_0^t ds \exp\left(  -\int_{s}^{t} d\xi a[f](x,\xi)\right)\int_0^{\infty}\int_0^{\infty} W(w,x,y,z)f(s,dy)f(s,dz)f(s,dx)\mathbbm{1}(w=y+z-x\geq0)\\
&+&\int_0^{\infty}\int_0^{\infty}\int_0^{\infty} e^{\beta (y+z-w)}\int_0^t ds \exp\left(  -\int_{s}^{t} d\xi a[f](x,\xi)\right) W(w,x,y,z) f(s,dw)f(s,dy)f(s,dz)\mathbbm{1}(R<x=y+z-w<R+1)\\
&+&\int_R^{R+1} e^{\beta x} dx\int_0^t ds \exp\left(  -\int_{s}^{t} d\xi a[f](x,\xi)\right)\int_0^{\infty}\int_0^{\infty} W(w,x,y,z) f(s,dy)f(s,dz)\mathbbm{1}(w=y+z-x\geq0)\\
&&=I_{1}^{R}[f]+I_{3}^{R}[f]+I_{2}^{R}[f].
\end{eqnarray*}

We will prove the theorem by first showing that the estimate in the theorem holds
for $I^{R}[f](t)$ for all $R$. We notice that the above quantity contains one
term , namely $I_{2}^{R}[f]$, which is quadratic in $f$ and two terms, $I_{3}^{R}[f]$ and $I_{1}^{R}[f]$, which are cubic in $f$.  Between the cubic terms, the term $I_3^{R}[f]$ is a bit more difficult to estimate, while the other one is simpler, so in the following computation we will estimate $I_3$ in more detail.  Also, we will estimate these terms separately for two cases, namely,
when $R\leq1$ and $R>1$. We need to make two more remarks before we proceed to
compute the estimates: 1) In the quantity $J[f](dx,s)$ the domain can be
described in terms of three regions, namely:
\begin{align*}
\Delta_{1}(x)  &  :=\{(y,z): 0\leq y<x, 0\leq z<x, w=y+z-x\geq0\}\\
\Delta_{2}(x)  &  := \{(y,z): z\geq x, 0\leq y\leq x\}\\
\Delta_{3}(x)  &  := \{(y,z): y\geq x, z\geq x\}
\end{align*}
The domain of integration contains these three regions plus the region in y-z
plane which is symmetric to $\Delta_{2}(x)$. In the computation that follows
we will denote by $I^{R}_{j}(\Delta_{k})$ the part of $I^{R}_{j}[f]$ evaluated
on region $\Delta_{k}(x)$ of the y-z plane, for $j=1,2$ and $k=1,2,3$. 2) We
use the following decomposition for the term $a[f](x,s)$:
\begin{align}\label{adecomp}
a[f](x,s)  &  = 2\sqrt{x}m(f)+S[f](x,s)\\
&  = 2\sqrt{x}m(f)+2\left(  x\int_{0}^{\infty}f(s,dw) G(x,w)+\int_{0}^{\infty}%
\int_{0}^{\infty} f(s,dw)f(s,dy)W(x,w,y,z)\mathbbm{1}(z=x+w-y\geq0)\right) \\ \nonumber
& = 2\sqrt{x}m(f) + S_1[f](x,s)+S_2[f](x,s), \nonumber
\end{align}
where $m(f)$ is the associated mass and,
\begin{align*}
G(x,w)  &  = \frac{1}{3}\left(  \frac{w}{x}\right)  ^{\frac{3}{2}},
\text{if}\quad w\leq x,\\
&  = \frac{1}{3}+\frac{w}{x}-\sqrt{\frac{w}{x}}, \text{if}\quad w\geq x
\end{align*}
This means that:
\begin{align*}
\exp\left(  -\int_{s}^{t} d\xi a[f](x,\xi)\right)  <\exp\left(  -2(t-s)\sqrt
{x}m(f)\right)
\end{align*}
In the computation below we have to keep the above inequality in mind and
control the integrals either by choosing a short time (if $R\leq1$) or by the
above inequality (for $R>1$). We suppress the dependence of the measure $f$ on time, writing explicitly only the energy variable. Also $C$ denotes a generic numerical constant
whose value may change from line to line. Let us now consider the different
cases:

\vspace{0.2in} \textsl{Case 1: $R\leq1$} Let us first estimate the
quadratic term:
\vspace{0.2in}

i)
\begin{align*}
I^{R}_{2}(\Delta_{1})  &  = \int_{R}^{R+1}dx\frac{1}{\sqrt{x}}\int_{0}^{t} ds
\exp\left(  -\int_{s}^{t} d\xi a[f](x,\xi)\right)  \int_{0}^{x}f(dz)\int
_{x-z}^{x} dy f(dy)e^{\beta(y+z)}(\sqrt{w}e^{-\beta w})\\
&  \leq\frac{C}{\sqrt{\beta}}t||f||_{G^{\beta}}^{2},
\end{align*}
ii)
\begin{align*}
I^{R}_{2}(\Delta_{2})  &  = \int_{R}^{R+1}dx\int_{0}^{t} ds \exp\left(
-\int_{s}^{t} d\xi a[f](x,\xi)\right)  \int_{x}^{\infty} \int_{0}^{x}
f(dy)f(dz)e^{\beta x}\sqrt{\frac{y}{x}}\\
&  \leq\int_{0}^{t} ds\int_{R}^{R+1}dx e^{\beta x}\int_{x}^{\infty} \int
_{0}^{x}  f(dy)f(dz)\\
&  \leq C ||f||_{G^{\beta}} \int_{0}^{t} ds\int_{R}^{R+1}dx e^{\beta x}
\sum_{k=x}^{\infty}\int_{k}^{k+1}f(dz)e^{\beta z}e^{-\beta z}\\
&  \leq C t ||f||_{G_{\beta}}^{2} \int_{R}^{R+1}dx e^{\beta x} \sum
_{k=x}^{\infty} e^{-\beta k}\\
&  = Ct||f||_{G^{\beta}}^{2}\frac{e^{\beta}}{e^{\beta}-1}\\
&  \leq Ct||f||_{G^{\beta}}^{2}%
\end{align*}
iii)
\begin{align*}
I^{R}_{2}(\Delta_{3})  &  = \int_{R}^{R+1}dx\int_{0}^{t} ds \exp\left(
-\int_{s}^{t} d\xi a[f](x,\xi)\right)  \int_{x}^{\infty} \int_{x}^{\infty}
f(dy)f(dz)e^{\beta x}\\
&  \leq\int_{0}^{t} ds\int_{R}^{R+1}dxe^{\beta x}\left(  \sum_{k=x}^{\infty
}\int_{k}^{k+1}f(dz)e^{\beta z}e^{-\beta z}\right)  ^{2}\\
&  \leq t||f||_{G^{\beta}}^{2}\int_{R}^{R+1}dxe^{\beta x}\left(  \sum
_{k=x}^{\infty} e^{-\beta k}\right)  ^{2}\\
&  \leq Ct||f||_{G^{\beta}}^{2}%
\end{align*}
Let us now compute an upper bound for the cubic term $I^R_3$ .

i)
\begin{align*}
I^{R}_{3}(\Delta_{1})  
&  \leq\int_{0}^{t}ds\int_{0}^{R+1}\int_{0}^{R+1}\int_{0}^{R+1}
f(dw)f(dy)f(dz)e^{\beta(y+z-w)}\\
&  \leq Ct||f||_{G^{\beta}}^{3}%
\end{align*}
ii)
\begin{align*}
I^{R}_{3}(\Delta_{2})  &  = \int_{R}^{\infty}\int_{0}^{R+1}\int_{max(y,y+z-R-1)}^{min(z,y+z-R)}\int_{0}^{t} ds
\exp\left(  -\int_{s}^{t} d\xi a[f](y+z-w,\xi)\right)  f(dw)f(dy)f(dz)\sqrt{\frac{y}{y+z-w}}e^{\beta (y+z-w)}\\
&  \leq\int_{0}^{t} ds \left[  \int_{R}^{R+1}f(dz) e^{\beta z}\int_{0}%
^{R+1}f(dy) e^{\beta y}\int_{min(0,y+z-(R+1))}^{y+z-R}f(dw)\right. \\
&  {}+\left.  e^{\beta(R+1)}\int_{R+1}^{\infty}f(dz)\int_{0}^{R+1}f(dy)\int
_{min(0,y+z-(R+1))}^{y+z-R}f(dw)\right] \\
&  \leq C\int_{0}^{t} ds||f||_{G^{\beta}}^{2}\left(  ||f||_{G^{\beta}%
}+e^{\beta(R+1)}\sum_{k=R+1}^{\infty}\int_{k}^{k+1}f(dz)e^{\beta z}e^{-\beta
z}\right) \\
&  \leq Ct||f||_{G^{\beta}}^{3}%
\end{align*}
iii)
\begin{align*}
I^{R}_{3}(\Delta_{3})  &
 \leq\int_{0}^{t} ds \int_{R}^{\infty}f(dz)\int_{R}^{\infty}
f(dy)\int_{y+z-(R+1)}^{y+z-R}f(dw)e^{\beta w}\\
&  \leq\int_{0}^{t} ds \left(  \sum_{k=R}^{\infty}\int_{k}^{k+1}
f(dy)e^{\beta y}e^{-\beta y}\right)  ^{2}||f||_{G^{\beta}}\\
&  \leq t \left(  \sum_{k=R}^{\infty}e^{-\beta k}\right)  ^{2}%
||f||_{G^{\beta}}^{3}\\
&  \leq Ct||f||_{G^{\beta}}^{3}%
\end{align*}
Putting all of the above estimates together we have that for $R\leq1$,
\begin{equation}
\label{smallR}I^{R}[f](t)\leq Ct ||f||_{G^{\beta}}^{2}%
(1+||f||_{G^{\beta}}).
\end{equation}
\textsl{Case 2: $R> 1$} We now consider the more delicate case, i.e., when
$R>1$. In our estimations in this case, besides controlling the time, we also
need to introduce a ``small'' parameter $\epsilon$ which we would fix in terms
of $\beta$ in the course of our computations. As before, we estimate the terms
in each of the three regions $\Delta_{j}$.

i)
\begin{align*}
&  I^{R}_{2}(\Delta_{1})\\
&  = \int_{R}^{R+1}dx\int_{0}^{t} ds \exp\left(  -\int_{s}^{t} d\xi
a[f](x,\xi)\right)  \int_{0}^{x} \int_{x-z}^{x} f(dy)f(dz)\sqrt{\frac{w}%
{x}}e^{\beta x}\mathbbm{1}(w=y+z-x\geq 0)\\
&  \leq\int_{R}^{R+1}dx\int_{0}^{t} ds e^{-2(t-s)\sqrt{x}m(f)}e^{\beta
x}\left[  \int_{0}^{\infty}\int_{0}^{\infty}\mathbbm{1}(0\leq
w=y+z-x\leq\epsilon)f(dy)f(dz)\sqrt{\frac{w}{x}}\right. \\
&  {}\left.  +\int_{\epsilon}^{x} \int_{x+\epsilon-z}^{x} f(dy)f(dz)\sqrt
{\frac{w}{x}} \mathbbm{1}(w=y+z-x>\epsilon)\right] \\
&  \leq\int_{R}^{R+1}dx\int_{0}^{t} ds e^{-2(t-s)\sqrt{x}m(f)}\frac{1}%
{\sqrt{x}}\left[  \sqrt{\epsilon}\int_{0}^{x+\epsilon}\int_{max(0,x-z)}%
^{x+\epsilon-z} f(dy)f(dz)e^{\beta(y+z-w)}\right. \\
&  {}\left.  +\int_{\epsilon}^{x} \int_{x+\epsilon-z}^{x}f(dy)f(dz)\sqrt
{w}e^{\beta(y+z-w)} \right] \\
&  \leq\sqrt{\epsilon}||f||_{G^{\beta}}\int_{R}^{R+1}dx\int_{0}^{t} ds
e^{-2(t-s)\sqrt{x}m(f)}\frac{1}{\sqrt{x}}\sum_{k=0}^{\lceil x+\epsilon\rceil
}\int_{k}^{k+1}f(dz)e^{\beta z}\\
&  +\int_{R}^{R+1}dx\int_{0}^{t} ds \frac{e^{-2(t-s)\sqrt{x}m(f)}}{\sqrt{x}%
}\sum_{k=0}^{\lceil x \rceil}\sum_{l=0}^{\lceil x \rceil}%
\mathbbm{1}(k+l>x+\epsilon)\int_{k}^{k+1}f(dz)e^{\beta z}\int_{l}^{l+1}
e^{\beta y}f(dy)\sqrt{w}e^{-\beta w}\\
&  \leq\sqrt{\epsilon}||f||_{G^{\beta}}^{2}\int_{R}^{R+1}dx\frac{\lceil
x+\epsilon\rceil}{\sqrt{x}}\int_{0}^{t} ds e^{-2(t-s)\sqrt{x}m(f)}\\
&  + ||f||_{G^{\beta}}^{2}\int_{R}^{R+1}dx\frac{1}{\sqrt{x}}\int_{0}^{t} ds
e^{-2(t-s)\sqrt{x}m(f)}\sum_{m=\epsilon}^{\epsilon+\lceil x \rceil}%
\sum_{s=-\lfloor x+m\rfloor}^{\lceil x+m \rceil}\sqrt{m}e^{-\beta m}\\
&  \leq C||f||_{G^{\beta}}^{2}\frac{1}{m(f)}\left(  \sqrt{\epsilon
}+e^{-\beta}\right)  ,
\end{align*}
where we choose $\epsilon=\beta^{-1}$ in order to ensure that the function
$g(x)=\sqrt{x}e^{-\beta x}$ is monotone decreasing on $(\epsilon,\infty)$.

ii)
\begin{align*}
I^{R}_{2}(\Delta_{2})  &  = \int_{R}^{R+1}dx\int_{0}^{t} ds \exp\left(
-\int_{s}^{t} d\xi a[f](x,\xi)\right)  \int_{x}^{\infty} \int_{0}^{x}
f(dy)f(dz)e^{\beta x}\sqrt{\frac{y}{x}}\\
&  \leq\int_{0}^{t} ds\int_{R}^{R+1}dx e^{\beta x}\int_{x}^{\infty} \int
_{0}^{x} f(dy)f(dz)\\
&  \leq C \int_{0}^{t} ds\int_{R}^{R+1}dx e^{\beta x} \sum_{k=x}^{\infty}%
\int_{k}^{k+1}f(dz)e^{\beta z}e^{-\beta z}\sum_{l=0}^{x}\int_{l}^{l+1}
f(dy)e^{\beta y}e^{-\beta y}\\
&  \leq Ct||f||_{G^{\beta}}^{2}%
\end{align*}
The computation of $I^{R}_{2}(\Delta_{3})$ in this case can be done just like
in the case $R\leq1$ with the result that
\[
I^{R}_{2}(\Delta_{3})\leq Ct||f||_{G^{\beta}}^{2}.
\]
We now estimate the cubic terms:
\begin{align*}
&  I^{R}_{3}(\Delta_{1})\\
&  =      \int_{0}^{R+1}\int_{0}^{R+1}\int_{max(0, y+z-R-1)}^{min(y,z,y+z-R)}e^{\beta (y+z-w)}\int_{0}^{t} ds \exp\left(  -\int_{s}^{t} d\xi
a[f](y+z-w,\xi)\right)  f(dw)f(dy)f(dz)\sqrt{\frac
{w}{y+z-x}}\\
&  \leq \int_{0}^{R+1}\int_{0}^{R+1}\int_{max(0,y+z-R-1)}^{min(y,z,y+z-R)}\int_{0}^t ds e^{-2(t-s)\sqrt{R}m(f)}f(dw)f(dy)f(dz)\sqrt{\frac{w}{R}}\left(
\mathbbm{1}(0\leq w\leq\epsilon)\right. \\
&  {}\left.  +\mathbbm{1}(w>\epsilon)\right) \\
\end{align*}
Now
\begin{align*}
&  \int_{0}^{R+1}\int_{0}^{R+1}\int_{max(0,y+z-R-1)}^{min(y,z,y+z-R)}\int_{0}^t ds e^{-2(t-s)\sqrt{R}m(f)}f(dw)f(dy)f(dz)\sqrt{\frac{w}{R}}
\mathbbm{1}(0\leq w\leq\epsilon) e^{\beta (y+z-w)}\\
&  \leq\int_{0}^{t} ds e^{-2(t-s)\sqrt{R}m(f)}\sqrt{\frac{\epsilon}{R}}\int_{0}^{R+1}\int_{0}^{\epsilon}\int_{max(0, w-z+R)}^{w-z+R+1} f(dy)f(dw)f(dz)e^{\beta (R+1)}\\
&  \leq\sqrt{\epsilon}||f||_{G^{\beta}}^{2} \frac{1}{2Rm(f)}  \sum_{k=0}^{\lceil R+1 \rceil}\int_{k}^{k+1}dz f(dz)e^{\beta
z}\\
&  \leq\frac{\sqrt{\epsilon}}{m(f)}||f||_{G^{\beta}}^{3}.
\end{align*}
The other part is estimated as:
\begin{align*}
& \int_{0}^{R+1}\int_{0}^{R+1}\int_{max(0,y+z-R-1)}^{min(y,z,y+z-R)}\int_{0}^t ds e^{-2(t-s)\sqrt{R}m(f)}f(dw)f(dy)f(dz)\sqrt{\frac{w}{R}} e^{\beta (y+z-w)}
\mathbbm{1}(w>\epsilon)\\
&  \leq\int_{0}^{t} ds e^{-2(t-s)\sqrt{R}m(f)}\frac{1}{\sqrt{R}}\int_{0}^{R+1}\int_{\epsilon}^{R+1}\int_{max(0, w-z+R)}^{w-z+R+1} \sqrt{w}f(dy)f(dw)f(dz)e^{\beta
(y+z-w)}\\
&  \leq\int_{0}^{t} ds e^{-2(t-s)\sqrt{R}m(f)}\frac{1}{\sqrt{R}}\int_{\epsilon}
^{R+1}\sqrt{w}f(dw)e^{-\beta w}\int_{0}^{R+1}f(dz) e^{\beta z}\int_{max(0,w+R-z)}^{w+R-z+1}f(dy)e^{\beta y}\\
&  \leq\int_{0}^{t} ds e^{-2(t-s)\sqrt{R}m(f)}\frac{1}{\sqrt{R}}||f||_{G^{\beta}
}\sum_{\epsilon}^{\epsilon+\lceil R+1 \rceil}\int_{k}^{k+1} f(dw)\sqrt
{w}e^{-\beta w}\sum_{l=0}^{\lceil R+1 \rceil}\int_{l}^{l+1}f(dz)e^{\beta
z}\\
&  \leq||f||_{G^{\beta}}^{3}\frac{1}{2Rm(f)}\lceil R+1 \rceil\sum_{k=\epsilon}^{\epsilon+1+\lceil R\rceil}\sqrt
{k}e^{-2\beta k}\\
&  \leq C||f||_{G^{\beta}}^{3}\frac{1}{m(f)}\left(  \sqrt{\epsilon
}e^{-2\beta\epsilon}+e^{-2\beta}\right)  .
\end{align*}
Thus:
\[
I^{R}_{3}(\Delta_{1})\leq C||f||_{G^{\beta}}^{3}\left( \frac{\sqrt{\epsilon}+e^{-2\beta}}{m(f)}\right)  .
\]
ii)
\begin{align*}
&  I^{R}_{3}(\Delta_{2})\\
&  =\int_{R}^{\infty}\int_{0}^{R+1}\int_{max(y,y+z-R-1)}^{min(z,y+z-R)} e^{\beta (y+z-w)}\int_{0}^{t} ds \exp\left(  -\int_{s}^{t} d\xi
a[f](y+z-w,\xi)\right)  f(dw)f(dy)f(dz)\sqrt
{\frac{y}{y+z-w}}\\
&  \leq\int_{R}^{R+1}\int_{0}^{R+1}\int_{max(y,y+z-R-1)}^{min(z,y+z-R)} e^{\beta (y+z-w)}\int_{0}^{t} ds e^{\left(  -2(t-s)\sqrt{R}m(f)\right)}  f(dw)f(dy)f(dz)\sqrt
{\frac{y}{y+z-w}}\\
&  {}+\int_{R+1}^{\infty}\int_{max(y,y+z-R-1)}^{min(z,y+z-R)} e^{\beta (y+z-w)}\int_{0}^{t} ds e^{\left(  -2(t-s)\sqrt{R}m(f)\right)} f(dw)f(dy)f(dz)\sqrt
{\frac{y}{y+z-w}}%
\end{align*}
Now:
\begin{align*}
& \int_{R}^{R+1}\int_{0}^{R+1}\int_{max(y,y+z-R-1)}^{min(z,y+z-R)} e^{\beta (y+z-w)}\int_{0}^{t} ds e^{\left(  -2(t-s)\sqrt{R}m(f)\right) } f(dw)f(dy)f(dz)\sqrt
{\frac{y}{y+z-w}}\\
&  \leq \int_{R}^{R+1}\int_{0}^{R+1}\int_{max(0,w-z+R)}^{min(w,w-z+R+1)} e^{\beta (y+z-w)}\int_{0}^{t} ds e^{\left(  -2(t-s)\sqrt{R}m(f)\right)}  f(dy)f(dw)f(dz)\\
&  \leq Ct||f||_{G^{\beta}}^{3}.
\end{align*}
The other part is evaluated as:
\begin{align*}
&  \int_{R+1}^{\infty}\int_0^{R+1}\int_{max(y,y+z-R-1)}^{min(z,y+z-R)} e^{\beta (y+z-w)}\int_{0}^{t} ds e^{\left(  -2(t-s)\sqrt{R}m(f)\right) } f(dw)f(dy)f(dz)\sqrt
{\frac{y}{y+z-w}}\\
&  \leq t\sqrt{\frac{\epsilon}{R}}\int_{R+1}^{\infty}\int_{0}^{\epsilon}%
\int_{max(y, y+z-R-1)}^{min(z, y+z-R)} f(dw)f(dy)f(dz)e^{\beta(R+1)}+\\
&  \quad+t \frac{e^{\beta(R+1)}}{\sqrt{R}}\int_{R+1}^{\infty}dz\int_{\epsilon}^{R+1}\int_{max(y,y+z-R-1)}^{min(z,y+z-R)}\sqrt{y}
f(dw)f(dy)f(dz)\\
&  \leq t\sqrt{\frac{\epsilon}{R}}||f||_{G^{\beta}}^{3} e^{\beta(R+1)}%
\sum_{k=R+1}^{\infty}e^{-\beta k}+\\
&  \quad+||f||_{G^{\beta}}\frac{e^{\beta(R+1)}}{\sqrt{R}}\sum_{k=R+1}^{\infty}\int_{k}^{k+1}dz f(z)\sum
_{l=\epsilon}^{\epsilon+1+\lceil R\rceil}\int_{l}^{l+1}dy f(y)\sqrt{y}\\
&  \leq Ct||f||_{G^{\beta}}^{3}\frac{1}{\sqrt{R}}\left(  \sqrt{\epsilon}+e^{-\beta
}\right)
\end{align*}
Finally the term $I^{R}_{3}(\Delta_{3})$ is evaluated in the same way as in
the case $R<1$. In all the above estimates we keep in mind the suppressed dependence of the measures $f$ on the time variable $s$. Putting the estimates together and using the fact that
$\epsilon= \beta^{-1}$  we obtain, after taking supremum over $R$ (since the right hand side of the above estimates do not depend upon $R$) and $s$, the following bound:
\begin{eqnarray*}
||\mathcal{T}[f]||_{G^{\beta}}(t)\leq  ||f_0||_{G^{\beta}} + C\left(  t+\frac{1}{m(f)\sqrt{\beta}}\right)
\left(\sup_{0\leq s\leq t}||f(s,.)||_{G^{\beta}}\right)^{2}\left(1+\sup_{0\leq s\leq t}||f(s,.)||_{G^{\beta}}\right)
\end{eqnarray*}

The bound above means that, in $G^{\beta}_{\kappa,T}$, given $f(0,.)$, such that $||f_0||_{G^{\beta}}<\kappa$, one can choose a sufficiently small time $T_*$ (depending on $\kappa$) and sufficiently large $\beta$  (which of course depends on $m(f)$) such that, for $T<T_{*}$ and $0\leq t\leq T$ :
$$||\mathcal{T}[f]||_{G^{\beta}}(t)\leq  ||f_0||_{G^{\beta}} + C\left( T+\frac{1}{m(f)\sqrt{\beta}}\right)\kappa^2 (1+\kappa)\leq \kappa $$i.e., the map $\mathcal{T}$ maps $G^{\beta}_{\kappa,T}$ to itself .

\end{proof}

We now prove the following lemma:
\begin{lemma} For $T<T_*$ the map $\mathcal{T}$ is weakly sequentially continuous in $G^{\beta}_{\kappa,T}$.
\end{lemma}
\begin{proof}
Let $\{f_n\}$ be a sequence in $G^{\beta}_{\kappa,T}$ such that $f_n\rightarrow f$ weakly in $G^{\beta}_{\kappa,T}$. We want to show that $\mathcal{T}$ is continuous, i.e., that we can pass to the limit in $\int_{0}^{\infty}\phi(x)\mathcal{T}[f_n](t,dx)$. We do this in several steps as follows:

First of all let us notice that:
\begin{eqnarray*}
a_n(s,x)&=& \int_{0}^{\infty}\int_0^{\infty} W(w,x,y,z)f_n(s,dw)f_n(s,dy)\mathbbm{1}(z=x+w-y\geq 0) \\
&{}&+ \int_{0}^{\infty}\int_0^{\infty} W(w,x,y,z)f_n(s,dw)f_n(s,dz)\mathbbm{1}(y=x+w-z\geq 0)\\
&{}&+ \int_{0}^{\infty}\int_0^{\infty} W(w,x,y,z)f_n(s,dw)dy\mathbbm{1}(z=x+w-y\geq 0)\\
&=& \int_{0}^{\infty}\int_0^{\infty} W(w,x,y,z)f_n(s,dw)f_n(s,dy)\mathbbm{1}(z=x+w-y\geq 0)\\
&{}&+\int_{0}^{\infty}\int_0^{\infty} W(w,x,y,z)f_n(s,dw)f_n(s,dz)\mathbbm{1}(y=x+w-z\geq 0)\\
&{}&+\int_{0}^{\infty}f_n(s,dw)\int_{-(x+w)}^{x+w}W(w,x,\frac{1}{2}(x+w+\xi),\frac{1}{2}(x+w-\xi))d\xi
\end{eqnarray*}

We then notice that for $x>0$, $W(w,x,y,z)$ is a bounded continuous function of  $w, y, z$. This means, by virtue of the weak convergence of $f_n$ to $f$ that:

$$\lim_{n\rightarrow\infty}\int_0^{\infty}W(w,x,y,z)f_n(s,dw)\mathbbm{1}(z=x+w-y\geq 0)=\int_0^{\infty}W(w,x,y,z)f(s,dw)\mathbbm{1}(z=x+w-y\geq 0),$$
and hence

$$\lim_{n\rightarrow\infty}\int_{0}^{\infty}f_n(s,dy)\int_0^{\infty}W(w,x,y,z)f_n(s,dw)\mathbbm{1}(z=x+w-y\geq 0)=\int_0^{\infty}f(s,dy)\int_0^{\infty}W(w,x,y,z)f(s,dw)\mathbbm{1}(z=x+w-y\geq0).$$

Similarly: $$\lim_{n\rightarrow\infty}\int_{0}^{\infty}\int_0^{\infty} W(w,x,y,z)f_n(s,dw)f_n(s,dz)\mathbbm{1}(y=x+w-z\geq 0)=\int_{0}^{\infty}\int_0^{\infty} W(w,x,y,z)f(s,dw)f(s,dz)\mathbbm{1}(y=x+w-z\geq 0).$$

For the limiting value of the last term, let us define:$$\int_{-(x+w)}^{x+w}W(w,x,\frac{1}{2}(x+w+\xi),\frac{1}{2}(x+w-\xi))d\xi=\Phi(w,x).$$ Then $\Phi(w,x)$ is a continuous unbounded function such that on any set $[L,\infty]$, with $L>0$ we have:
$$\int_L^{\infty}f_n(s,dw)\Phi(w,x)\leq 2\int_L^{\infty}(x+w)f_n(s,dw)\leq C\left(xe^{-\beta L} + \frac{1}{\sqrt{L}}E(f_n)\right),$$
where $E(f_n)$ denotes the energy associated with $f_n$ and is finite. This means the integral whose limiting value we are interested in, can be made arbitrarily small on any $[L,\infty]$ by choosing a large enough $L$, while convergence is guaranteed on compact sets $(0,L)$. Thus:
$$\lim_{n\rightarrow\infty}\int_0^{\infty}f_n(s,dw)\Phi(w,x)=\int_0^{\infty}f(s,dw)\Phi(w,x).$$

The above computations then mean that:
$$\lim_{n\rightarrow\infty} a_n(s,x)=a(s,x)>0,$$ where
\begin{eqnarray*}
a(s,x)&=&\int_{0}^{\infty}\int_0^{\infty} W(w,x,y,z)f(s,dw)f(s,dy)\mathbbm{1}(z=x+w-y\geq 0) \\
&{}&+ \int_{0}^{\infty}\int_0^{\infty} W(w,x,y,z)f(s,dw)f(s,dz)\mathbbm{1}(y=x+w-z\geq 0)\\
&{}&+ \int_{0}^{\infty}\int_0^{\infty} W(w,x,y,z)f(s,dw)dy\mathbbm{1}(z=x+w-y\geq 0).
\end{eqnarray*}
Also, $$\lim_{n\rightarrow\infty}\Psi_n(s,x)=\Psi(s,x),$$ where $$\Psi_n(s,x)=e^{-\int_0^{t}a_n(s,x)ds}.$$
In the following computations we will encounter the term below:
$$\frac{\Psi_n(t,x)}{\Psi_n(s,x)}=e^{-\int_s^t a_n(\xi,x)d\xi}<1.$$

We now focus our attention on:
$$\mathcal{I}_n=\int_{0}^{\infty}\phi(x)\int_0^t ds \frac{\Psi_n(t,x)}{\Psi_n(s,x)} J[f_n](s,dx). $$
This consists of two cubic and one quadratic terms in $f_n$. We look at the relevant limits below:
\begin{eqnarray*}
\mathcal{I}_{1,n}&=&\int_0^t\int_0^{\infty}\int_0^{\infty}\int_0^{\infty}\frac{\Psi_n(t,x)}{\Psi_n(s,x)}\phi(x)W(w,x,y,z)f_n(s,dx)f_n(s,dy)f_n(s,dz)ds\\
&=& \int_0^t\int_0^{\infty}\int_0^{\infty}\int_0^{\infty}\frac{\Psi(t,x)}{\Psi(s,x)}\phi(x)W(w,x,y,z)f(s,dx)f_n(s,dy)f_n(s,dz)ds\\
&{}&+\int_0^t\int_0^{\infty}\int_0^{\infty}\int_0^{\infty}\left[\frac{\Psi_n(t,x)}{\Psi_n(s,x)}f_n(s,dx)-\frac{\Psi(t,x)}{\Psi(s,x)}f(s,dx)\right] W(w,x,y,z)\phi(x)f_n(s,dy)f_n(s,dz),
\end{eqnarray*}

where we write the last term as follows:
\begin{eqnarray*}
\Delta\mathcal{I}_{1,n} &=& \int_0^t\int_0^{\infty}\int_0^{\infty}\int_0^{\infty}\left[\frac{\Psi_n(t,x)}{\Psi_n(s,x)}f_n(s,dx)-\frac{\Psi(t,x)}{\Psi(s,x)}f_n(s,dx)\right.\\
&{}&\left.+\frac{\Psi(t,x)}{\Psi(s,x)}f_n(s,dx)-\frac{\Psi(t,x)}{\Psi(s,x)}f(s,dx)\right] W(w,x,y,z)\phi(x)f_n(s,dy)f_n(s,dz)\\
&{}&\leq\int_0^t\int_0^{\infty}\int_0^{\infty}\left(\int_0^{\infty}|\frac{\Psi_n(t,x)}{\Psi_n(s,x)}-\frac{\Psi(t,x)}{\Psi(s,x)} |f_n(s,dx)W(w,x,y,z)\phi(x)\right)f_n(s,dy)f_n(s,dz)\\
&{}&+ \int_0^t\int_0^{\infty}\int_0^{\infty}\left(|\int_0^{\infty}\frac{\Psi(t,x)}{\Psi(s,x)}W(w,x,y,z)\phi(x)f_n(s,dx)-\int_0^{\infty}\frac{\Psi(t,x)}{\Psi(s,x)}W(w,x,y,z)\phi(x)f(s,dx)|\right)f_n(s,dy)f_n(s,dz)\\
&{}&=\Delta\mathcal{I}^1_{1,n}+\Delta\mathcal{I}^2_{1,n}.
\end{eqnarray*}

We will first estimate the error term $\Delta\mathcal{I}^2_{1,n}$.
Let $$\zeta_{t,s}(w,x,y,z)=\frac{\Psi(t,x)}{\Psi(s,x)}W(w,x,y,z)\phi(x).$$ Then for $x>0$ $\zeta_{t,s}(w,x,y,z)$ is a continuous bounded function of its arguments and hence:

$$\lim_{n\rightarrow\infty}\int_0^{\infty}\zeta_{t,s}(w,x,y,z)f_n(s,dx)=\int_0^{\infty}\zeta_{t,s}(w,x,y,z)f(s,dx),$$
which means that $\Delta\mathcal{I}^2_{1,n}$ can be made arbitrarily small by choosing $n$ large enough.

We now turn to the other error term $\Delta\mathcal{I}^1_{1,n}.$ Let us first note that $$\frac{\Psi_n(t,x)}{\Psi_n(s,x)}\rightarrow \frac{\Psi(t,x)}{\Psi(s,x)},$$
as $n\rightarrow\infty$ point wise by dominated convergence.
We can now use Egorov's theorem, which asserts that for $\epsilon>0$ there is a measurable subset $B$ of $\mathbb{R}_{+}$ such that $f_n(B)<\epsilon$ and $\frac{\Psi_n(t,x)}{\Psi_n(s,x)}$
converges to $\frac{\Psi(t,x)}{\Psi(s,x)}$ uniformly on $\mathbb{R}_{+} \setminus B.$

Then:
\begin{eqnarray*}
\Delta\mathcal{I}^1_{1,n}&=&\int_0^t\int_0^{\infty}\int_0^{\infty}\left(\int_{\mathbb{R}_{+}\setminus B}|\frac{\Psi_n(t,x)}{\Psi_n(s,x)}-\frac{\Psi(t,x)}{\Psi(s,x)} |f_n(s,dx)W(w,x,y,z)\phi(x)\right)f_n(s,dy)f_n(s,dz)\\
&{}&+ \int_0^t\int_0^{\infty}\int_0^{\infty}\left(\int_{ B}|\frac{\Psi_n(t,x)}{\Psi_n(s,x)}-\frac{\Psi(t,x)}{\Psi(s,x)} |f_n(s,dx)W(w,x,y,z)\phi(x)\right)f_n(s,dy)f_n(s,dz)\\
&{}&\leq \int_0^t\int_0^{\infty}\int_0^{\infty}\left(\int_{\mathbb{R}_{+}\setminus B}|\frac{\Psi_n(t,x)}{\Psi_n(s,x)}-\frac{\Psi(t,x)}{\Psi(s,x)} |f_n(s,dx)W(w,x,y,z)\phi(x)\right)f_n(s,dy)f_n(s,dz)\\
&{}& C\epsilon,
\end{eqnarray*}
where C is some positive constant. Clearly this can be made arbitrarily small just by choosing $n$ large enough and $\epsilon$ small enough.
We thus conclude that the term $\Delta\mathcal{I}_{1,n}$ can be made as small as we wish. Let us now turn to $\mathcal{I}_n$ again.
\begin{eqnarray*}
\mathcal{I}_{1,n}&=&\Delta\mathcal{I}_{1,n}+\int_0^t\int_0^{\infty}\int_0^{\infty}\left[\int_0^{\infty}\frac{\Psi(t,x)}{\Psi(s,x)} W(w,x,y,z)f(s,dx)\right] f_n(s,dy)f_n(s,dz)ds\\
&{}&=\Delta\mathcal{I}_{1,n}+ \int_0^t ds\int_0^{\infty}f_n(s,dy)\int_0^{\infty}f_n(s,dz) F(t,s;y,z),
\end{eqnarray*}
where $$F(t,s;y,z)=\int_0^{\infty}\frac{\Psi(t,x)}{\Psi(s,x)} W(w,x,y,z)f(s,dx),$$ is a continuous and bounded function of $s,y,z.$ Thus:
\begin{eqnarray*}
\mathcal{I}_{1,n}&=& \Delta\mathcal{I}_{1,n}+\int_0^t ds\int_0^{\infty}f_n(s,dy)\int_0^{\infty}(f_n(s,dz)-f(s,dz)) F(t,s;y,z) + \int_0^t ds\int_0^{\infty}f_n(s,dy)\int_0^{\infty}f(s,dz) F(t,s;y,z)\\
&{}&= \Delta\mathcal{I}_{1,n}+\delta\mathcal{I}_{n}+ \int_0^t ds\int_0^{\infty}f_n(s,dy)\int_0^{\infty}f(s,dz) F(t,s;y,z),
\end{eqnarray*}
where $$\delta\mathcal{I}_{n}=\int_0^t ds\int_0^{\infty}f_n(s,dy)\int_0^{\infty}(f_n(s,dz)-f(s,dz)) F(t,s;y,z).$$
Evidently both $\delta\mathcal{I}_{n}$ and $\Delta\mathcal{I}_{1,n}$ can be made arbitrarily small while
$$\lim_{n\rightarrow\infty}\int_0^t ds\int_0^{\infty}f_n(s,dy)\int_0^{\infty}f(s,dz) F(t,s;y,z)=\int_0^t ds\int_0^{\infty}f(s,dy)\int_0^{\infty}f(s,dz) F(t,s;y,z).$$
Thus we finally have:

$$\lim_{n\rightarrow\infty}\mathcal{I}_{1,n}=\int_0^t ds\int_0^{\infty}f(s,dy)\int_{0}^{\infty}f(s,dz)\int_{0}^{\infty}\frac{\Psi(t,x)}{\Psi(s,x)}W(w,x,y,z)f(s,dx)=\mathcal{I}_1.$$
For the other cubic term can be written as follows:

\begin{eqnarray*}
\mathcal{I}_{2,n}&=& \int_0^t\int_0^{\infty}\int_0^{\infty}\int_0^{\infty} \frac{\Psi_n(t,x)}{\Psi_n(s,x)} W(w,x,y,z)\phi(x)f_n(s,dw)f_n(s,dy)f_n(s,dz)\mathbbm{1}(x=y+z-w\geq0)\\
&=& \int_0^t\int_0^{\infty}\int_0^{\infty}\left(\int_0^{y+z} \frac{\Psi_n(t,y+z-w)}{\Psi_n(s,y+z-w)} \phi(y+z-w)W(w,y+z-w,y,z)f_n(s,dw)\right) f_n(s,dy)f_n(s,dz).
\end{eqnarray*}

This is then evaluated in the same manner as $\mathcal{I}_{1,n}$ and the quadratic term is similarly dealt with, proving that the operator $\mathcal{T}$ is indeed continuous in the metric space $G^{\beta}_{\kappa,T_{*}}.$

\end{proof}



\begin{lemma}
Suppose $||f_0||_{G^{\beta}}<\kappa$ and $T<T_{*}$ is as in lemma $3$. The operator $\mathcal{T}$ is weakly compact in $G^{\beta}_{\kappa, T}.$
\end{lemma}

\begin{proof}
For any $f\in G^{\beta}_{\kappa, T}$ and any continuous bounded function $\phi$, let us define the following function:
$$A_{\phi}[f](t)=\int_{0}^{\infty}\phi(x)\mathcal{T}[f](t,dx).$$
We now show that this function is Lipshitz continuous in time. Differentiating both sides we have:
\begin{eqnarray}\label{compact_bound}
\lvert \frac{\partial}{\partial t}A_{\phi}[f](t) \rvert&\leq& \int_0^{\infty}|a(x,t)|\exp\left(-\int_0^t a(x,s)ds\right) f(0,dx)\phi(x)+\\ \nonumber
&+&\int_0^{\infty}\phi(x)\int_0^t ds |a(x,t)|\exp\left(-\int_s^t a(x,\xi)d\xi\right) J[f](s,dx) + \int_0^{\infty}\phi(x)J[f](t,dx).
\end{eqnarray}
We now use: $$|a(x,t)|\leq 2\sqrt{x}m(f)+2x\int_0^x \frac{1}{3}\left(\frac{w}{x}\right)^{\frac{3}{2}}f(t,dw)+2x\int_x^{\infty}\left(\frac{1}{3}+\frac{w}{x}\right)f(t,dw)+2\int_0^{\infty}\int_0^{\infty}W(w,x,y,z)f(t,dw)f(t,dy),$$
and the bounds achieved in lemma $3$ to obtain that the right hand side of \eqref{compact_bound} is bounded for $0\leq t \leq T$.

Thus $A_{\phi}[f](t)$ is Lipshitz continuous and since the Lipshitz bound holds for any $f\in G^{\beta}_{\kappa, T}$, we can conclude that the family $\mathcal{T}[f](t,.)$ is equicontinuous.
Thus the map $\mathcal{T}$ is weakly compact.

\end{proof}

Lemmas $3$, $4$ and $5$ then imply the following local existence theorem:
\begin{theorem}
Let $f_0\in G^{\beta}$. Then there exists  $\kappa>||f_0||_{G^{\beta}}$ and a time $T>0$ depending on $||f_0||_{G^{\beta}}$ such that there exists, for $0\leq t\leq T$, a solution $f(t,.)\in G^{\beta}_{\kappa, T}$ of \eqref{BN} in the sense of \eqref{MMsoln} and \eqref{mesmildsoln}.
\end{theorem}
\begin{proof}
This is just a consequence of the fact that the time evolution operator $\mathcal{T}:G^{\beta}_{\kappa, T}\rightarrow G^{\beta}_{\kappa, T}$, since it is  weakly sequentially continuous  and weakly compact , must have a fixed point $f\in G^{\beta}_{\kappa, T}$ by the Schauder-Tychonoff theorem. It is then clear that this $f$ must satisfy \eqref{MMsoln} and \eqref{mesmildsoln} by virtue of the definition \eqref{evolnop}.
\end{proof}
\vspace{1in}

We will now prove that the obtained measure-valued solution conserves mass and energy.  For any $f\in G^{\beta}_{\kappa, T}$ let us define an associated density: $$g(., dx)=\sqrt{x}f(.,dx).$$ Then the mass and energy functionals associated with the measure $f(.,dx)$ are :
$$m(f)=\int_{\mathbb{R}_{+}} g(.,dx),$$ and $$e(f)=\int_{\mathbb{R}_{+}} x g(.,dx).$$
We now prove the following lemma:
\begin{lemma}
Let $f\in G^{\beta}_{\kappa, T}$ be a measure-valued solution of \eqref{BN} in the sense of \eqref{MMsoln} and \eqref{mesmildsoln}. Let $g$ be the associated density. Then:
\begin{eqnarray*}
\frac{\partial}{\partial t}\left(\int_{\mathbb{R}_{+}} g(t, dx)\right)=\frac{\partial}{\partial t}\left(\int_{\mathbb{R}_{+}}x g(t, dx)\right)=0,\qquad a.e.\quad t\in[0,T]
\end{eqnarray*}
\end{lemma}
\begin{proof}
We will prove this lemma in two steps.

First let us note the equation below, for any $\phi\in C([0,T]\times[0,\infty))$, which follows from the definition of $g$:
\begin{eqnarray*}
&&\int_{0}^{\infty}g(t,dx)\phi(x,t) \\&=& \int_{0}^{\infty}\phi(x,t)\Psi(t,x)g(0,dx)+\int_0^{\infty}\int_0^{\infty}\int_0^{\infty}\int_0^t ds \left(\frac{\Psi(t,x)}{\Psi(s,x)}\right)\phi(x,t)\Phi(w,x,y,z)\frac{g(s,dy)g(s,dz)dx}{\sqrt{yz}}\mathbbm{1}(w=y+z-x\geq 0)\\
&+& \int_0^{\infty}\int_0^{\infty}\int_0^{\infty}\int_0^t ds \left(\frac{\Psi(t,x)}{\Psi(s,x)}\right)\phi(x,t)\Phi(w,x,y,z)\frac{g(s,dx)g(s,dy)g(s,dz)}{\sqrt{xyz}}\mathbbm{1}(w=y+z-x\geq 0)\\
&+& \int_0^{\infty}\int_0^{\infty}\int_0^{\infty}\int_0^t ds \left(\frac{\Psi(t,x)}{\Psi(s,x)}\right)\phi(x,t)\Phi(w,x,y,z)\frac{g(s,dw)g(s,dy)g(s,dz)}{\sqrt{wyz}}\mathbbm{1}(x=y+z-w\geq 0),
\end{eqnarray*}
where $\Phi(w,x,y,z)=\min(\sqrt{w},\sqrt{x},\sqrt{y},\sqrt{z})$ and $$\Psi(s,x)=e^{-\int_0^{t}a(s,x)ds}.$$
Since $a(t,x)$ is bounded on compact sets in $[0,T]\times[0,\infty)$, we can differentiate both sides with respect to $t$ and obtain:
\begin{eqnarray*}
\frac{\partial}{\partial t}\left(\int_0^{\infty}\phi(x,t) g(t, dx)\right)&=&\int_0^{\infty}g(t,dx)\frac{\partial}{\partial t}\phi(x,t) - \int_0^{\infty}\phi(x,t)a(x,t)g(t,dx)\\
&+&\int_0^{\infty}\int_0^{\infty}\int_0^{\infty} \phi(x,t)\Phi(w,x,y,z)\frac{g(t,dy)g(t,dz)dx}{\sqrt{yz}}\mathbbm{1}(w=y+z-x\geq 0) \\
&+&\int_0^{\infty}\int_0^{\infty}\int_0^{\infty}\phi(x,t)\Phi(w,x,y,z)\frac{g(t,dx)g(t,dy)g(t,dz)}{\sqrt{xyz}}\mathbbm{1}(w=y+z-x\geq 0)\\
&+& \int_0^{\infty}\int_0^{\infty}\int_0^{\infty}\phi(x,t)\Phi(w,x,y,z)\frac{g(t,dw)g(t,dy)g(t,dz)}{\sqrt{wyz}}\mathbbm{1}(x=y+z-w\geq 0)
\end{eqnarray*}
We now plug into the above equation the definition \eqref{adef} for $a(x,t)$ and observe that in the resulting equation the cubic and the quadratic terms can be collected and combined as follows, after a change of variables:
\begin{eqnarray*}
&&\int_0^{\infty}\int_0^{\infty}\int_0^{\infty} \phi(x,t)\Phi(w,x,y,z)\frac{g(t,dx)g(t,dy)g(t,dz)}{\sqrt{xyz}}\mathbbm{1}(w=y+z-x\geq 0)\\
&&+ \int_0^{\infty}\int_0^{\infty}\int_0^{\infty} \phi(x,t)\Phi(w,x,y,z)\frac{g(t,dw)g(t,dy)g(t,dz)}{\sqrt{wyz}}\mathbbm{1}(x=y+z-w\geq 0)\\
&&- \int_0^{\infty}\int_0^{\infty}\int_0^{\infty} \phi(x,t)\Phi(w,x,y,z)\frac{g(t,dx)g(t,dw)g(t,dy)}{\sqrt{wxy}}\mathbbm{1}(z=x+w-y\geq 0)\\
&&-\int_0^{\infty}\int_0^{\infty}\int_0^{\infty} \phi(x,t)\Phi(w,x,y,z)\frac{g(t,dw)g(t,dx)g(t,dz)}{\sqrt{wxz}}\mathbbm{1}(y=x+w-z\geq 0)\\
&=&\int_0^{\infty}\int_0^{\infty}\int_0^{\infty}\left[\phi(w,t)+\phi(x+y-w,t)-2\phi(x,t)\right]\Phi(w,x,y,z)\frac{g(t,dx)g(t,dw)g(t,dy)}{\sqrt{wxy}}\mathbbm{1}(x+y=z+w).
\end{eqnarray*}
Similarly, for the quadratic terms:
\begin{eqnarray*}
&&\int_0^{\infty}\int_0^{\infty}\int_0^{\infty}\phi(x,t)\Phi(w,x,y,z)\frac{g(t,dy)g(t,dz)dx}{\sqrt{yz}}\mathbbm{1}(w=y+z-x\geq 0)\\
&& - \int_0^{\infty}\int_0^{\infty}\int_0^{\infty}\phi(x,t)\Phi(w,x,y,z)\frac{g(t,dx)g(t,dw)dz}{\sqrt{wx}}\mathbbm{1}(y=x+w-z\geq 0)\\
&=&\int_0^{\infty}\int_0^{\infty}\int_0^{\infty}\left[\phi(w,t)+\phi(x+y-w,t)-2\phi(x,t)\right]\Phi(w,x,y,z)\frac{g(t,dx)g(t,dw)g(t,dy)}{\sqrt{wxy}}\mathbbm{1}(x+y=z+w)
\end{eqnarray*}

Thus:
\begin{eqnarray}\label{wtevoln}
&&\frac{\partial}{\partial t}\left(\int_0^{\infty}\phi(x,t) g(t, dx)\right)=\int_0^{\infty}g(t,dx)\frac{\partial}{\partial t}\phi(x,t)+\\ \nonumber
&+&\int_0^{\infty}\int_0^{\infty}\int_0^{\infty}\left[\phi(w,t)+\phi(x+y-w,t)-2\phi(x,t)\right]\Phi(w,x,y,z)\frac{g(t,dx)g(t,dw)g(t,dy)}{\sqrt{wxy}}\mathbbm{1}(x+y=z+w)\\ \nonumber
&+&\int_0^{\infty}\int_0^{\infty}\int_0^{\infty}\left[\phi(w,t)+\phi(x+y-w,t)-2\phi(x,t)\right]\Phi(w,x,y,z)\frac{g(t,dx)g(t,dw)g(t,dy)}{\sqrt{wxy}}\mathbbm{1}(x+y=z+w)\\ \nonumber
&=& \int_0^{\infty}g(t,dx)\frac{\partial}{\partial t}\phi(x,t) + G_3[\phi]+G_2[\phi] ,
\end{eqnarray}

where $G_3$ and $G_2$ denote the cubic and quadratic terms in $g$ respectively.
We now apply equation \eqref{wtevoln} to a sequence of test functions defined as follows:
$$\phi_n(x) = x\zeta_n(x),$$ where,
\begin{eqnarray*}
\zeta_n(x) &=& 1,\quad x\leq n\\
&=& 0,\quad x\geq n+1\\
\text{and}\qquad \zeta'_{n}&\leq& 0.
\end{eqnarray*}

Then it easy to see that the finiteness of the energy integral implies:
$$\lim_{n\rightarrow\infty}G_3[\phi_n]=G_3[\tilde{\phi}]=0,$$
where $\tilde{\phi}(x)=x.$

To conclude that the integral $G_2[\phi_n]$ is convergent as well, it is enough to notice the following estimates:
\begin{eqnarray*}
&&\int_0^{\infty}\Phi(w,x,y,z)\left[\phi_n(w,t)+\phi_n(x+y-w,t)-2\phi(x,t)\right]\mathbbm{1}(x+y=w+z)\\
&\leq& \min(\sqrt{x},\sqrt{y})\int_0^{x+y}dw\left[\phi_n(w,t)+\phi_n(x+y-w,t)-2\phi(x,t)\right]\\
&\leq& \min(\sqrt{x},\sqrt{y})(x+y)^2,
\end{eqnarray*}
which in turn implies:
\begin{eqnarray*}
G_2[\phi_n]\leq C\int_0^{\infty}y^{\frac{3}{2}}g(t,dy)\int_0^y g(t,dx).
\end{eqnarray*}
The right hand side is finite, since the moments of $f\in G^{\beta}_{\kappa}$ are finite.
Thus one can the limit of the quadratic term as:
$$\lim_{n\rightarrow\infty}G_2[\phi_n]=G_2[\tilde{\phi}]=0,$$
as before.

Thus equation \eqref{wtevoln} implies, when we evaluate it for the test function $\phi_n$ and take the limit:
\begin{eqnarray*}
\frac{\partial}{\partial t}\left(\int_0^{\infty}\tilde{\phi}(x,t) g(t, dx)\right)=0\qquad a.e.\quad t\in[0,T],
\end{eqnarray*}
where $\tilde{\phi}(x)=x.$

This proves the conservation of energy. The conservation of mass follows by a similar argument.

\end{proof}

\textbf{Acknowledgment} This work has been supported by the CRC 1060 ``The Mathematics of Emergent Effects" at the University of Bonn, that is funded through the German Science Foundation (DFG).

\end{document}